\documentclass{article}
\usepackage[utf8]{inputenc}

\usepackage{amsmath}
\usepackage{stylesheet}

\usepackage{enumitem}

\newcommand{\D}{\mathcal{D}}
\newcommand{\E}{\mathbb{E}}
\newcommand{\Naive}{\mathsf{Naive}}
\newcommand{\Linear}{\mathsf{Linear}}
\newcommand{\LinearFixed}{\mathsf{LinearFixed}}
\newcommand{\IgnoreLarge}{\mathsf{IgnoreLarge}}
\newcommand{\IgnoreLargeExp}{\mathsf{IgnoreLargeExp}}
\newcommand{\Opt}{\mathsf{Opt}}
\newcommand{\Norm}{\mathcal{N}}
\newcommand{\Event}{\mathcal{E}}
\newcommand{\R}{\mathbb{R}}

\newcommand{\std}{\boldsymbol{\sigma}}
\newcommand{\observ}{\boldsymbol{y}}
\newcommand{\good}[2]{\mathcal{S}_{(#1, #2)}}
\newcommand{\goodMHR}[2]{\mathcal{S}^\textsc{MHR}_{(#1, #2)}}
\newcommand{\bad}[2]{\mathcal{L}_{(#1, #2)}}
\newcommand{\badMHR}[2]{\mathcal{M}_{(#1, #2)}}

\DeclareMathOperator*{\argmax}{arg\,max}

\DeclareMathOperator\erf{erf}
\usepackage{cleveref}

\title{Reward Selection with Noisy Observations}

\author{
    Kamyar Azizzadenesheli \\ 
    \small{Nvidia}
    \and Trung Dang \\ 
    \small{Purdue University}
    \and Aranyak Mehta \\
    \small{Google}
    \and Alexandros Psomas \\ 
    \small{Purdue University}
    \and Qian Zhang \\ 
    \small{Purdue University}
}

\date{}

\begin{document}

\maketitle

\begin{abstract}
We study a fundamental problem in optimization under uncertainty. There are $n$ boxes; each box $i$ contains a hidden reward $x_i$. Rewards are drawn i.i.d. from an unknown distribution $\D$. For each box $i$, we see $y_i$, an unbiased estimate of its reward, which is drawn from a Normal distribution with known standard deviation $\sigma_i$ (and an unknown mean $x_i$). Our task is to select a single box, with the goal of maximizing our reward. This problem captures a wide range of applications, e.g. ad auctions, where the hidden reward is the click-through rate of an ad. Previous work in this model~\cite{bax2012comparing} proves that the naive policy, which selects the box with the largest estimate $y_i$, is suboptimal, and suggests a linear policy, which selects the box $i$ with the largest $y_i - c \cdot \sigma_i$, for some $c > 0$. However, no formal guarantees are given about the performance of either policy (e.g., whether their expected reward is within some factor of the optimal policy's reward).

In this work, we prove that both the naive policy and the linear policy are arbitrarily bad compared to the optimal policy, even when $\D$ is well-behaved, e.g. has monotone hazard rate (MHR), and even under a ``small tail'' condition, which requires that not too many boxes have arbitrarily large noise.
On the flip side, we propose a simple threshold policy that gives a constant approximation to the reward of a prophet (who knows the realized values $x_1, \dots, x_n$) under the same ``small tail'' condition. We prove that when this condition is not satisfied, even an optimal clairvoyant policy (that knows $\D$) cannot get a constant approximation to the prophet, even for MHR distributions, implying that our threshold policy is optimal against the prophet benchmark, up to constants. En route to proving our results, we show a strong concentration result for the maximum of $n$ i.i.d. samples from an MHR random variable that might be of independent interest.
\end{abstract}

\newpage

\section{Introduction}

Suppose that you are given $n$ boxes, with box $i$ containing a hidden reward $x_i$. Rewards are drawn independently and identically distributed (i.i.d.) from an unknown distribution $\D$. For each box $i$, you see an unbiased estimate $y_i$ of its reward: nature draws noise $\epsilon_i \sim \Norm(0,\sigma_i)$ with known $\sigma_i$, and you observe $y_i = x_i + \epsilon_i$. Your goal is to select the box with the highest reward $x_i$. This fundamental problem, originally introduced by Bax et al.~\cite{bax2012comparing}, captures a wide range of applications. The original motivation of Bax et al.~\cite{bax2012comparing} is ad auctions, where one can think of the hidden reward $x_i$ as the click-through rate of an ad, and the observed value $y_i$ as an estimation of the click-through rate produced by a machine learning algorithm; these algorithms typically have different amounts of data, and therefore different variance in the error, across different populations. 

If the distribution $\D$ is known, the optimal policy simply calculates the posterior expectation $R_i(y_i) = \E[ X_i \mid Y_i = y_i ]$ for each box $i$ and selects the box with the largest $R_i(y_i)$. However, when $\D$ is not known, this calculation is, of course, not possible. Furthermore, if $\epsilon_i$s were drawn i.i.d. (that is, if all $\sigma_i$s were equal), it should be intuitive that $\Naive$, the policy that picks the box with the largest observation $y_i$, is optimal, since $R_i(y_i) = \E[ X_i \mid X_i + \epsilon_i = y_i ]$ ``should'' be a monotone non-decreasing function of $y_i$.\footnote{As we show in one of our technical lemmas, this happens to be true when $\epsilon_i$ is drawn from $\Norm(0,\sigma)$, but, perhaps surprisingly, this is not true for an arbitrary noise distribution. To see this, consider the case that $X_i$ is uniform in the set $\{ -1 , +1 \}$ and $\epsilon_i$ is uniform in the set $\{ -10, +10 \}$. In this case, $\E[ X_i \mid X_i + \epsilon_i = -9 ] = 1 > -1 = \E[ X_i \mid X_i + \epsilon_i = 9 ]$.}

Bax et al.~\cite{bax2012comparing} show that $\Naive$ is suboptimal when the $\sigma_i$s are not equal. Specifically, they consider a family of \emph{linear} policies. A linear policy with parameter $c$ selects the box with the largest $y_i - c \cdot \sigma_i$; for $c=0$ we recover $\Naive$. Bax et al.~\cite{bax2012comparing} show that the derivative of the expected reward is strictly positive at $c=0$; that is, the $\Naive$ policy is not optimal, even within the family of linear policies. However, and this brings us to our interest here, no other formal guarantees are given. Is the best linear policy, or even the $\Naive$ policy, a good (e.g. constant) approximation to the optimal policy? Are there better policies, outside the family of linear policies?

\subsection{Our contribution}

 Without loss of generality, we assume that $\std = (\sigma_1, \dots, \sigma_n)$ satisfies  $\sigma_1 \leq ... \leq \sigma_n$. Naturally, if $\sigma_i$ is extremely large for almost all $i$, no policy, including a clairvoyant policy that knows $\D$, can hope to achieve any non-trivial performance guarantees (e.g., perform better than picking  a random box). We start by making this intuition precise.
Informally, given $\D$, $n$ and $c$, $\std$ has large noise if $\sigma_{n^c}$ is at least $\tilde{\Omega}( \E[\D_{n^c:n^c}] )$.\footnote{Recall that $\D_{k:n}$ is the $k$-th lowest of  $n$ i.i.d. samples from $\D$.} Under this condition, we show that, even for the case of a distribution $\D$ with monotone hazard rate (MHR),\footnote{A distribution has monotone hazard rate (MHR) if $\frac{1-F(x)}{f(x)}$ is a non-increasing function.} an optimal clairvoyant policy (which knows $\D$) cannot compete with $\E[\D_{n:n}]$, the expected reward of a prophet that knows the rewards $x_1, \dots, x_n$. Despite the fact that the prophet is a very strong benchmark, we note that, as we see later in the paper, our policies compete against the prophet, in similarly ``noisy'' environments. We further show that, assuming a bit more noise, $\sigma_{cn} \in \tilde{\Omega}(\E[\D_{cn:cn}])$ for $cn \in O(1)$, an optimal clairvoyant policy has reward comparable to the reward of picking a box uniformly at random. See~\Cref{sec: prelims} for the precise definitions, and~\Cref{sec: lower bounds for large noise} for the formal statements and proofs. We henceforth assume that the environment has ``small noise.''

We proceed to analyze the performance of known policies under this assumption. In~\Cref{subsec: naive fails} we study the $\Naive$ policy, which selects the box with the highest reward, and show that not only is it suboptimal, but that it can be made suboptimal for \emph{every} distribution $\D$ (\Cref{thm:master-lower-bound-s}). Specifically, given an arbitrary distribution $\D$, there exist choices for $n$ and $\std$ (satisfying the aforementioned ``small noise'' assumption) such that the optimal (non-clairvoyant) policy has reward at least $\E[ \D_{n:n}]/2$, while the $\Naive$ policy has a reward of at most $4\E[\D]$.  Our construction has a small number, $\Theta(\log(n))$, boxes with large noise, with the remaining boxes having no noise. The intuition is that, with high probability, a random large noise box is chosen by $\Naive$, while picking among the no noise boxes yields a reward of almost $\E[\D_{n:n}]$. Selecting $\D$ such that $\E[\D_{n:n}] \in \Theta( n\E[\D] )$, we have that  $\Naive$ provides only a trivial approximation to the optimal reward.  

In~\Cref{subsec: lower bound linear} we study linear policies. Surprisingly, this family of policies can also be made suboptimal in a similarly strong way. Given an arbitrary MHR distribution $\D$, there exist choices for $n$ and $\std$ (again, satisfying the aforementioned ``small noise'' assumption) such that the optimal policy has reward at least a constant times $\E[ \D_{n:n}]$, but no linear policy can get expected reward more than a constant times $\E[\D]$ (\Cref{thr:master-lower-bound}). By letting $\D$ be the exponential distribution, we get a lower bound of $\Omega(\log(n))$ for the approximation ratio of linear policies. 
Constructing a counter-example for linear policies is more delicate. First, observe that on all $\std$'s and realizations $y$'s, every linear policy's performance is at most the best $\LinearFixed_c$ policy, which discounts all boxes by a weight $c$ tailored to $\std$ and $y$. For a fixed and small $c$, a construction similar to the one for $\Naive$ works. For a fixed and large $c$, $\LinearFixed_c$ ``over-discounts'', and therefore a construction with many small noise boxes (that are not picked with high probability) works. We show how to combine these two ideas into a single construction where all  $\LinearFixed_c$ policies fail with high probability, and then use a union bound to relate to the best linear policy.

Combined, Theorems~\ref{thm:master-lower-bound-s} and~\ref{thr:master-lower-bound} show that, even if we know that $\D$ belongs to the (arguably very well-behaved) family of monotone hazard rate distributions, we need a new approach. En route to showing~\Cref{thr:master-lower-bound}, we prove a lemma about the concentration of the maximum of $n$ i.i.d. samples from an MHR distribution which might be of independent interest. It is known that order-statistics of MHR distributions also satisfy the MHR condition~\cite{barlow1996}. Furthermore, MHR distributions exceed their mean with probability at least $1/e$. Therefore, $\Pr\left[ \D_{n:n} \geq \E[\D_{n:n}] \right] \geq 1/e$. Here, we show that $\D_{n:n}$ does not exceed twice its mean with high probability (\Cref{lem:lower bound on max passing expectation}): $\Pr\left[ \D_{n:n} \leq 2\E[\D_{n:n}] \right] \geq 1 - \frac{1}{n^{3/5}}$, implying a very small tail for $\D_{n:n}$. The proof of this result is based on a new lemma (which again might be of independent interest) which states that the $(1-1/n)$-quantile value of an MHR distribution $\D$ is within a constant of $\E[\D_{n:n}]$.

At a high level, the downfall of both $\Naive$ and linear policies is that they treat very different types of boxes in a virtually identical manner: $\Naive$ does not take in the noise information at all, while linear policies utilize this information in a very crude way, and discount boxes with massively different order of noises using the same weight. Intuitively, a good policy should identify large noise boxes and ignore them. However, a non-trivial obstacle, is that a noise being ``large'' is relative to $\D$, which is unknown. 

In~\Cref{sec:positive} we propose our new policy, that circumvents this issue. The policy is quite simple: pick $\alpha \sim U[0,1]$, and run $\Naive$ on the $\alpha$ fraction of the boxes with the lowest noise (i.e. boxes $1$ through $\alpha n$). Therefore, if, e.g. a constant fraction of the boxes has small noise, we have a constant probability of keeping a constant fraction of them.
In more detail, if a $c$ fraction of the boxes has low noise, and specifically, if $\sigma_{cn} \leq \frac{\E[\D_{cn:cn}]}{5\sqrt{2\ln(n)}}$ (arguably, a very permissive bound), then our policy gives a $\frac{c^2}{20}$ approximation to $\E[\D_{n:n}]$, the expected reward of a prophet. Clearly, if $c$ is a constant, we get a constant approximation. Interestingly, our policy provides the same guarantees even in a setting with a lot less information, where the $\sigma_i$s are \emph{unknown}, and only their order is available to the policy. For the case of MHR distributions we further improve this result. The policy itself has a slight twist: pick $\alpha \sim U[0,1]$, and run $\Naive$ on the $n^{\alpha}$ boxes with the lowest noise (i.e. boxes $1$ through $n^{\alpha}$). This time, if $n^c$ boxes have low noise, and specifically if $\sigma_{n^c} \leq \frac{\E[\D_{n^c:n^c}]}{18\sqrt{2\ln(n^c)}}$, this version of our policy guarantees a $c^2/576$ approximation to the prophet. For a constant $c$, our approximation to the prophet is again a constant, and we only require $n^c$ boxes with bounded noise.

\subsection{Related Work}

\cite{bax2012comparing}, whose contribution we already discussed, and~\cite{mahdian2022regret}, are the two works most closely related to ours.~\cite{mahdian2022regret} study a very similar model to ours, where the reward $x_i$ for each box $i$ is not stochastic, but adversarial, and the noise distribution is not $\Norm(0,\sigma_i)$, but an arbitrary (known) zero-mean distribution $A_i$.~\cite{mahdian2022regret} are interested in finding policies with small worst-case \emph{regret}, defined as the difference between the maximum reward and the expected performance of the policy, where the expectation is over only the random noise. A policy is then a constant approximation if its regret is within a constant of the optimal regret; in contrast, for us, a policy is a constant approximation if its expected \emph{reward} is within a constant of the expected reward of the optimal policy/a prophet.~\cite{mahdian2022regret} show that in their model as well, the naive policy which picks the box with the highest observation $y_i$ is arbitrarily bad (in terms of regret) even in the $n=2$ case. Similar to our results here,~\cite{mahdian2022regret} show that there is a function $\theta$ from random variables to positive reals, such that picking the box with the largest $y_i - \theta(A_i)$ is a constant approximation (in terms of regret) to the optimal policy. Note that, in the case of our policy, this function is especially simple: $\theta(A_i) = 0$ if $\sigma_i$ is small, otherwise $\theta(A_i)$ is infinite.

A phenomenon related to the naive policy being suboptimal, both in the model studied here/the model of \cite{bax2012comparing}, as well as the model of~\cite{mahdian2022regret}, is the winner's curse~\cite{thaler1988anomalies}, where multiple bidders, with the same ex-post value for an item, estimate this value independently and submit bids based on those estimates; the winner tends to have a bid that's an overestimate of the true value. Our problem is also related to robust optimization which studies optimization in which we seek solutions that are robust with respect to the realization of uncertainty; see~\cite{bertsimas2011theory} for a survey. Finally, there has been a lot of work on the related problem of finding the maximum (or the top $k$ elements) given noisy information, see, e.g.,~\cite{feige1994computing,braverman2016parallel,braverman2019sorted,cohen2020instance}.

Many of our theorems can be strengthened by additionally assuming that $\D$ is MHR. MHR distributions are known to satisfy a number of interesting properties, see~\cite{barlow1996} for a textbook. In algorithmic economics, such properties have been exploited to enable strong positive results for a number of problems, including the sample complexity of revenue maximization~\cite{dhangwatnotai2010revenue,cole2014sample,huang2015making,guo2019settling,guo2021robust}, the competition complexity of dynamic auctions~\cite{liu2018competition}, and the design of optimal and approximately optimal ~\cite{hartline2009simple,daskalakis2012symmetries,cai2011extreme,allouah2020prior,giannakopoulos2021optimal}.

\section{Preliminaries}\label{sec: prelims}

There are $n$ boxes. The $i$-th box contains a reward $x_i$. These rewards are drawn i.i.d. from an \textit{unknown} distribution $\D$ with a cumulative distribution function $F$ and density function $f$. We assume that $\D$ is supported on $[0,\infty)$. Rewards are not observed by our algorithm.
Instead, nature draws unbiased estimates, $y_1, \dots, y_n$, where $y_i$ is drawn from a normal distribution with (an unknown) mean $x_i$ and a \textit{known} standard deviation $\sigma_i$. We refer to $y_i$ as the $i$-th observation. We often write $X_i$ and $Y_i$ for the random variable for the $i$-th reward and $i$-th observation, respectively. Note that $Y_i$ can be equivalently thought as $Y_i = X_i + \epsilon_i$, where the noise $\epsilon_i$ is drawn from $\Norm(0,\sigma_i)$. Our goal is to select a single box $i$ with the goal of maximizing the (expected) realized reward.

\paragraph{Policies and expected rewards}
Formally, a policy $A$ maps the public information, the pair $(\std, \observ)$, $\std = (\sigma_1, \dots, \sigma_n)$ and $\observ = (y_1, \dots, y_n)$, to a distribution over boxes. We write $R_{A}(\D, \std, \observ)$ for the expected reward of a policy $A$ under true reward distribution $\D$ and observations $\observ = (y_1, \dots, y_n)$, where the standard deviation of the noise is according to $\std = (\sigma_1, \dots, \sigma_n)$, and where this expectation is with respect to the randomness of $A$ and the randomness in the rewards. In order to evaluate a policy under a fixed reward distribution $\D$ we need to take an additional expectation over the random observations $\observ = (y_1, \dots, y_n)$. We overload notation and write $R_{A}(\D, \std) = \E_{\observ} \left[ R_{A}(\D, \std, \observ) \right]$ for the expected reward of a policy $A$ under true reward distribution $\D$, where the standard deviation of the noise is according to $\std = (\sigma_1, \dots, \sigma_n)$.

\paragraph{Previous policies and benchmarks}
\cite{bax2012comparing} consider two simple policies. The $\Naive$ policy always selects the box $i$ with the largest observation $y_i$.  A linear policy $\Linear_\gamma$, parameterized by a function $\gamma : \R^n \times \R^n \to \R$, chooses the box $i$ which maximizes $y_i - \gamma(\std,\observ) \cdot \sigma_i$. 

We use the following two policies as useful benchmarks: the \emph{optimal policy}, and the \emph{prophet}.
The optimal policy for a distribution $\D$, $\Opt_\D$, selects the box $i$ with maximum $\E \left[ X_i \mid Y_i = y_i \right]$. Its expected reward in outcome $\observ$ is precisely $\max_i \E \left[ X_i \mid Y_i = y_i \right]$. That is, $R_{\Opt_{\D}}(\D, \std) = \E_{\observ} \left[ \max_{i \in [n]} \E \left[ X_i \mid Y_i = y_i \right] \right]$. 
Finally, the (expected) reward of a prophet who knows $x_1, \dots, x_n$, for a distribution $\D$, is equal to $\E[ \D_{n:n} ]$, the expected maximum of $n$ i.i.d. draws from $\D$.




\paragraph{Formalizing ``small'' and ``large'' noise environments}

Clearly, if $\sigma_i$ is large for almost all $i \in [n]$, then no policy can hope to get a non-trivial guarantee. Therefore, we intuitively need a condition that captures the fact that we need small noise for enough boxes. In the following couple of definitions, we formalize precisely what we mean by ``small'' and ``enough''.

\begin{definition}[Small noise]\label{dfn: small noise}
For any distribution $\D$, any $n$ and any $c \in (0,1]$, let $\good{\D}{n,c}$ be the set of vectors $\std \in \R_+^n$ where at least $cn$ values in $\std$ are at most $\frac{\E[\D_{cn:cn}]}{5\sqrt{2\ln n}}$. Formally, $\good{\D}{n,c} = \{\std \in \R_+^n \mid \sigma_1 \le \dots \le \sigma_n \text{ and } \sigma_{cn} \le \frac{\E[\D_{cn:cn}]}{5\sqrt{2\ln n}}\}$.
\end{definition}

For the case of MHR distributions, we only need a weaker condition to guarantee strong positive results. We state this condition in~\Cref{dfn:small noise mhr}. 

\begin{definition}[Small noise for MHR]\label{dfn:small noise mhr}
For any MHR distribution $\D$ and any $n$, let $\goodMHR{\D}{n,c}$ be the set of vectors $\std \in \R_+^n$ where at least $n^c$ values in $\std$ are at most $\frac{\E[\D_{n^c:n^c}]}{18\sqrt{2c \ln n}}$. Formally, $\goodMHR{\D}{n,c} = \{\std \in \R_+^n \mid \sigma_1 \le \dots \le \sigma_n \text{ and }\sigma_{n^c} \le \frac{\E[\D_{n^c:n^c}]}{18\sqrt{2c \ln n}}\}$.
\end{definition}

Ideally, we would like to, whenever $\std \notin \good{\D}{n,c}$ or $\goodMHR{\D}{n,c}$, have strong negative results for, say, the optimal policy. We show such strong negative results for the optimal \emph{clairvoyant} policy, even for MHR distributions, even under a condition close to $\goodMHR{\D}{n,c}$. On the negative side, the precise condition is not the complement of $\goodMHR{\D}{n,c}$, but we lose an extra $\sqrt{c}$ factor. Under the following ``medium noise'' condition, we cannot hope to compete against the prophet (\Cref{thm:opt-bad-vs-prophet}).

\begin{definition}[Medium noise]
\label{dfn: large noise mhr}
For any distribution $\D$, any $n$ and any $c \in (0, 1]$, let $\badMHR{\D}{n,c}$ be the set of vectors $\std \in \R_+^n$ where at most $n^c$ values in $\std$ is at most $\frac{\E[\D_{n^c:n^c}]\}}{18 c\sqrt{2\ln n}}$. Formally, $\badMHR{\D}{n,c} = \{\std \in \R_+^n \mid \sigma_1 \le \dots \le \sigma_n \text{ and }\sigma_{n^c} > \frac{\E[\D_{n^c:n^c}]\}}{18 c\sqrt{2\ln n}}$.
\end{definition}

Finally, under the following ``large noise'' condition, closer to the complement of $\good{\D}{n,c}$ (with an extra $\sqrt{\ln(n)/\ln(cn)}$ factor), we cannot hope to do better than picking a box uniformly at random (\Cref{thm: opt bad vs random}).

\begin{definition}[Large noise]\label{dfn: large noise}
For any distribution $\D$, any $n$ and any $c \in (0,1]$, let $\bad{\D}{n,c}$ be the set of vectors $\std \in \R_+^n$ where at most $cn$ values in $\std$ are at most $\frac{\E[\D_{cn:cn}] \cdot \sqrt{\ln n}}{ln(cn)}$. Formally, $\bad{\D}{n,c} = \{\std \in \R_+^n \mid \sigma_1 \le \dots \le \sigma_n \text{ and } \sigma_{cn} > \frac{\E[\D_{cn:cn}] \cdot \sqrt{\ln n}}{\ln(cn)} \}$.
\end{definition}

\subsection{Technical Lemmas}\label{subsec:technical}

Here, we present some definitions and a few technical lemmas that will be useful throughout the paper. All missing proofs can be found in Appendix~\ref{app: missing from technical}.

We often use the following lemma (\Cref{lem:normal-bound}) about the CDF of the standard normal distribution, and a lemma (\Cref{lem:compare-order-stats}) about the relation between the expected maximum of $a$ and $b$ i.i.d. samples from an arbitrary distribution $\D$. We write $\D_{k:n}$ for the $k$-th lowest order statistic out of $n$ i.i.d. samples, that is, $\D_{1:n} \leq \D_{2:n} \leq \dots \leq \D_{n:n}$.
Throughout the paper, $\Phi(x)$ is the CDF of the standard normal distribution, and $\phi(x)$ is the PDF of the standard normal distribution.

\begin{lemma}[\cite{gordon1941}]
\label{lem:normal-bound}
For all $t > 0$, we have $1 - \frac{1}{\sqrt{2 \pi}} \frac{1}{t} e^{-t^2/2} \le \Phi(t) \le 1 - \frac{1}{\sqrt{2 \pi}} \frac{t}{t^2+1} e^{-t^2/2}$. Furthermore, this implies directly that for all $t > 0$, $1 - \frac{\phi(t)}{t} \le \Phi(t) \le 1 - \frac{t \phi(t)}{t^2 + 1}$.
\end{lemma}

\begin{lemma}
\label{lem:compare-order-stats}
    For any distribution $\D$ supported on $[0, \infty)$ and for any two integers $1 \le a < b$, we have
    $\frac{\E[\D_{a:a}]}{a} \ge \frac{\E[\D_{b:b}]}{b}.$
\end{lemma}

The following definitions will be crucial in describing our lower bounds.

\begin{definition}[Cai and Daskalakis~\cite{cai2011extreme}]\label{dfn: alpha_m}
For a distribution $\D$, let $\alpha^{(\D)}_m = \inf\{x \mid F(x) \ge 1 - \frac{1}{m} \}$ be the $(1 - \frac{1}{m})$-th quantile of $\D$.
\end{definition}

\begin{definition}\label{dfn: beta_m}
For a distribution $\D$, let $\beta^{(\D)}_m = \inf\{x \mid \E[\D \mid \D \ge x] \cdot \Pr[\D \ge x] \le \frac{\E[\D]}{m}\}$ be the smallest threshold such that the contribution to $\E[\D]$ from values at least this threshold is at most $\frac{\E[\D]}{m}$.
\end{definition}


\paragraph{Technical lemmas for MHR distributions}

Here, we prove a technical lemma for the concentration of the maximum of $n$ i.i.d. samples of an MHR distribution, that might be of independent interest. 

It is known that the maximum of i.i.d. draws from an MHR distribution is also MHR~\cite{barlow1996}. This implies that the probability that the maximum exceeds its mean, $\Pr \left[ \D_{n:n} \geq  \E[\D_{n:n}] \right]$, is at least $1/e$. In~\Cref{lem:lower bound on max passing expectation} we show that, in fact, this maximum concentrates around its mean: it does not exceed twice its mean with high probability. We note that a related, but incomparable, statement is given by~\cite{cai2011extreme}, who show that at least a $(1-\epsilon)$-fraction of $\E[\max_i X_i]$ is contributed by values no larger than $\E[\max_i X_i] \cdot \log(\frac{1}{\epsilon})$, where the $X_i$s are (possibly not identical) MHR distributions.

\begin{lemma}\label{lem:lower bound on max passing expectation}
    For any MHR distribution $\D$ and any $n \ge 4$,
    we have
    \[\Pr[\D_{n:n} < 2 \cdot \E[\D_{n:n}]] \ge 1 - \frac{1}{n^{3/5}}.\]
\end{lemma}

\Cref{lem:lower bound on max passing expectation} is an immediate consequence of the following two lemmas. The first is shown in~\cite{cai2011extreme}; the second we prove in~\Cref{app: missing from technical}.

\begin{lemma}[\cite{cai2011extreme}; Lemma 34]\label{lem: cai daskalakis alpha bound}
If the distribution of a random variable $X$ satisfies MHR, $m \geq 1$ and $d \geq 1$, then $d \alpha^\text{(X)}_m \geq \alpha^\text{(X)}_{m^d}$.
\end{lemma}

\begin{lemma}
\label{lem:order-stat-vs-quantile}
    For any MHR distribution $\D$ and any $n \ge 4$, we have $\frac{1}{3} \cdot \E[\D_{n:n}] \le \alpha^{(\D)}_n \le \frac{5}{4} \cdot \E[\D_{n:n}]$.
\end{lemma}

\begin{proof}[Proof of \Cref{lem:lower bound on max passing expectation}]

Together the lemmas give that
$\alpha^{(\D)}_{n^{8/5}} \le^{\text{(\Cref{lem: cai daskalakis alpha bound})}} \frac{8}{5} \alpha^{(\D)}_n \le^{\text{({\Cref{lem:order-stat-vs-quantile}})}} 2 \E[\D_{n:n}]$. Therefore,
\[
    \Pr[\D_{n:n} \le 2 \E[\D_{n:n}]] \ge \Pr[\D_{n:n} \le \alpha^{(\D)}_{n^{8/5}}] 
        = \left(1 - \frac{1}{n^{8/5}}\right)^n 
        \ge^{\text{(Bernoulli's inequality)}} 1 - \frac{1}{n^{3/5}}.\qedhere
\]
\end{proof}

\section{Negative results for large noise environments}\label{sec: lower bounds for large noise}

Before discussing small noise environments, we show strong lower bounds for the optimal clairvoyant policy (an optimal policy that knows $\D$) in large noise environments, even under the assumption that the distribution $\D$ is MHR. All missing proofs can be found in~\Cref{app: missing from lower bounds for regimes}.

Starting with ``medium'' noise,~\Cref{thm:opt-bad-vs-prophet} shows that, for an MHR distribution $\D$, when $\std \in \badMHR{\D}{n,c}$, even an optimal clairvoyant policy cannot approximate the prophet to some absolute value proportional to $\sqrt{c}$. First, as we discussed in~\Cref{sec: prelims}, note that $\goodMHR{\D}{n,c}$ is \emph{almost}, but not exactly, the complement of $\badMHR{\D}{n,c}$; the complement of $\badMHR{\D}{n,c}$ includes $\std$ where at least $n^c$ values in $\std$ are at most $\frac{\E[\D_{n^c:n^c}]\}}{18 c\sqrt{2\ln n}}$, while $\goodMHR{\D}{n,c}$ is characterized by $\std$'s containing at least $n^c$ values upper bounded by $\frac{ \E[\D_{n^c:n^c}]}{18\sqrt{2c \ln n}}$, implying that $\badMHR{\D}{n,c}$ is a strict subset of the complement of $\goodMHR{\D}{n,c}$ as $c \le 1$. This leaves a gap (arguably insignificant, but a gap nonetheless) in our understanding.
On the flip side, our negative result holds against the (well-behaved) class of MHR distributions, even against the strong benchmark of the optimal clairvoyant policy.

\begin{theorem}
\label{thm:opt-bad-vs-prophet}
    There exists a MHR distribution $\D$ where $\E[\D_{k:k}] \in \omega(\E[D])$ for $k \in \omega(1)$, such that for all $n \ge n_0$, for some constant $n_0$, for all $c \in [\frac{1}{400 \sqrt{\ln n}}, 1]$, and all $\std \in \badMHR{\D}{n, c}$, we have
     \[R_{\Opt_\D}(\D, \std) \in O \left( \sqrt{c} \cdot \E[\D_{n:n}] \right).\]
\end{theorem}

One way to interpret~\Cref{thm:opt-bad-vs-prophet} is that, for any desired constant approximation $\alpha$, for all large enough $n$, one can select a small enough $c$ and $\std$ that satisfies the ``medium'' noise condition (noting that this condition also depends on $c$), such that the optimal clairvoyant policy does not achieve an $\alpha$ approximation. We include the fact that $\E[\D_{k:k}] \in \omega(\E[D])$, to highlight that the distribution is not trivial. For example, it is not the case that the expectation is already a constant away from the expected maximum. 

The distribution that witnesses \Cref{thm:opt-bad-vs-prophet} is the standard half-normal distribution $\D = |\Norm(0, 1)|$. We start by proving that this distribution is MHR, and bounding its expected maximum value. 

\begin{lemma}\label{lem:order-stats-half-norm}
    $\D = |\Norm(0, 1)|$ is MHR, $\E[\D] = \sqrt{\frac{2}{\pi}}$, and $\frac{4}{5} \cdot \sqrt{\ln n} \le \E[\D_{n:n}] \le 3 \sqrt{2} \cdot \sqrt{\ln n}$ for $n \ge 8$.
\end{lemma}

Since order statistics are preserved under affine transformations, an immediate corollary is the following.

\begin{corollary}
\label{lem:order-stats-half-norm-general}
For all $\sigma > 0$, $\frac{4}{5} \cdot \sigma \sqrt{\ln n} \le \E\left[|\Norm(0, \sigma^2)|\right]_{n:n} \le 3 \sqrt{2} \cdot \sigma\sqrt{\ln n}$ for $n \ge 8$.
\end{corollary}

Towards bounding the optimal policy, we can compute the exact expression for $\E[X_i \mid Y_i = y_i]$.

\begin{lemma}
\label{lem:expected-posterior-half-norm}
Given $Y_i = X_i + \epsilon_i$ where $X_i \sim \D$ and $\epsilon_i \sim \Norm(0, \sigma_i^2)$, we have
\[\E[X_i \mid Y_i = y_i] = \frac{y_i}{\sigma_i^2 + 1} + \frac{\phi\left(\frac{-y_i}{\sigma_i \sqrt{\sigma_i^2 + 1}}\right)}{1 - \Phi\left(\frac{-y_i}{\sigma_i \sqrt{\sigma_i^2 + 1}}\right)} \cdot \frac{\sigma_i}{\sqrt{\sigma_i^2 + 1}}.\]
\end{lemma}

Unfortunately, while this form is exact, it is not easy to work with. We instead consider the following upper bound on $\E[X_i \mid Y_i = y_i]$.

\begin{lemma}
\label{lem:upp-bound-expected-posterior}
    Let $U_\sigma(y) = \sqrt\frac{2}{\pi} + \max\left\{0, \frac{y}{\sigma^2+1}\right\}$, then $\E[X_i \mid Y_i = y_i] \le U_{\sigma_i}(y_i)$ for all $\sigma_i$ and $y_i$.
\end{lemma}

\begin{proof}
We first consider the case where $y_i \ge 0$. In this case, $U_{\sigma_i}(y_i) = \sqrt\frac{2}{\pi} + \frac{y_i}{\sigma^2_i+1}$. Observe that $\phi(x) \le \frac{1}{\sqrt{2 \pi}}$ for all $x$, $1 - \Phi(x) \ge \frac{1}{2}$ for all $x \le 0$, and $\frac{\sigma_i}{\sqrt{\sigma_i^2 + 1}} \le 1$ for all $\sigma_i \ge 0$. Therefore, 
\[\E[X_i \mid Y_i = y_i] = \frac{y_i}{\sigma_i^2 + 1} + \frac{\phi\left(\frac{-y_i}{\sigma_i \sqrt{\sigma_i^2 + 1}}\right)}{1 - \Phi\left(\frac{-y_i}{\sigma_i \sqrt{\sigma_i^2 + 1}}\right)} \cdot \frac{\sigma_i}{\sqrt{\sigma_i^2 + 1}} \le \frac{y_i}{\sigma_i^2 + 1} + \sqrt\frac{2}{\pi} = U_{\sigma_i}(y_i).\]

If $y_i < 0$, we use the property that $\E[X_i \mid Y_i = y_i] \le \E[X_i \mid Y_i = 0]$ (this is due to the monotonicity of $\E[X_i \mid Y_i = y_i]$; see \Cref{lem: monotonicity of expected posterior}): $\E[X_i \mid Y_i = y_i] \le \E[X_i \mid Y_i = 0] \le U_{\sigma_i}(0) = U_{\sigma_i}(y_i)$.
\end{proof}

We are now ready to prove \Cref{thm:opt-bad-vs-prophet}.

\begin{proof}[Proof of~\Cref{thm:opt-bad-vs-prophet}]
Let $\D = |\Norm(0, 1)|$, and consider $\std = \in \badMHR{\D}{n, c}$ where, without loss of generality, we have $\sigma_1 \le \sigma_2 \le \dots \le \sigma_n$. This means that $\sigma_{n^c} > \frac{\E[\D_{n^c:n^c}]\}}{18 c\sqrt{2\ln n}} \geq^\text{(\Cref{lem:order-stats-half-norm})} \frac{4\sqrt{\ln(n^c)}}{90c \sqrt{2\ln n}} = \frac{\sqrt{2}}{45 \sqrt{c}}$. Note that the expected reward of the optimal policy is at most the expected reward of the optimal policy that picks $2$ boxes $u$ and $v$ where $u \in [1, n^c - 1]$ and $v \in [n^c, n]$, and \emph{then enjoys the rewards of both boxes}. 
The expected reward from choosing box $u$ is at most $\E[\max_{i \in [1, n^c - 1]} x_i] \le \E[\D_{n^c:n^c}]$. The expected reward from choosing box $v$ is at most the expected reward of $\Opt_{\D}$, restricted to choosing boxes from $n^c$ to $n$, which in turn is at most $\max_{i \in [n^c, n]} \E[X_i \mid Y_i = y_i]$. Therefore, the expected reward from box $v$ is upper bounded by:

\begin{align*}
    \E_{\observ}\left[\max_{i \in [n^c, n]} \E[X_i \mid Y_i = y_i]\right] &\le^\text{(\Cref{lem:upp-bound-expected-posterior})} \E_{\observ}\left[\max_{i \in [n^c, n]} U_{\sigma_i}(y_i)\right] \\
        &= \E\left[\max_{i \in [n^c, n]} U_{\sigma_i}\left(X_i + \Norm(0, \sigma_i^2)\right)\right] \\
        &\le^\text{($U_{\sigma_i}(y)$ is monotone)} \E\left[\max_{i \in [n^c, n]} U_{\sigma_i}\left(X_i + |\Norm(0, \sigma_i^2)|\right)\right] \\
        &= \E\left[\max_{i \in [n^c, n]} \sqrt\frac{2}{\pi} + \frac{\left(X_i + |\Norm(0, \sigma_i^2)|\right)}{\sigma_i^2 + 1}\right] \\
        &\le \E\left[\sqrt\frac{2}{\pi} + \max_{i \in [n^c, n]} \frac{X_i}{\sigma^2_{i}} + \max_{i \in [n^c, n]} \frac{|\Norm(0, \sigma_i^2)|}{\sigma_i^2}\right] \\
        &\le \sqrt\frac{2}{\pi} + \frac{\E\left[ |\Norm(0, 1)|_{n:n} \right]}{\sigma^2_{n^c}}  + \E\left[ \max_{i \in [n^c, n]} \left|\Norm\left(0, \frac{1}{\sigma_i^2}\right)\right| \right] \\
        &\leq^\text{(\Cref{lem:order-stats-half-norm-general})} \sqrt\frac{2}{\pi} + \frac{ 3\sqrt{2} \cdot \sqrt{\ln n}}{\sigma_{n^c}^2} + \frac{1}{\sigma_{n^c}} \cdot 3\sqrt{2} \cdot \sqrt{\ln n} \\
        &\le^{\left( \sigma_{n^c} > \frac{\sqrt{2}}{45 \sqrt{c}}\right)} \sqrt\frac{2}{\pi} + \frac{6075 \cdot c \sqrt{\ln n}}{\sqrt{2}} + 135 \sqrt{c \ln n} \\
        &\le^{\left(\frac{1}{400 \sqrt{\ln n}} \le c \le 1\right)} 4497 \sqrt{c \ln n}
\end{align*}

Combining, we have that $R_{\Opt_\D}(\D, \std) \le 4500 \sqrt{c \ln n} + \E[\D_{n^c:n^c}] \leq^\text{(\Cref{lem:order-stats-half-norm})} 4497 \sqrt{c \ln n} + \frac{3}{\sqrt{2}} \sqrt{c \ln n} \le 4500 \sqrt{c \ln n}$. Using~\Cref{lem:order-stats-half-norm}, this is at most $4500\sqrt{c} \cdot \frac{5}{4} \E[\D_{n:n}] \leq 5625 \sqrt{c} \cdot \E[\D_{n:n}]$.
\end{proof}

The following theorem shows that, if the environment has ``large'' noise, then the optimal clairvoyant policy is comparable to the policy that picks a random box.

\begin{theorem}\label{thm: opt bad vs random}
There exists an MHR distribution $\D$ where $\E[\D_{k:k}] \in \omega(\E[D])$ for $k \in \omega(1)$, such that for all $n \ge n_0$, for some constant $n_0$, for all $c \in [1/n,1]$, all and all $\std \in \bad{\D}{n, c}$, we have
\[ R_{\Opt_\D}(\D, \std) \in O\left( \sqrt{\ln(cn)} \E[\D] \right). \]
\end{theorem}

One way to interpret this theorem is that, given any constant target ratio $\alpha$ and any large enough $n$, one can pick $c$ small enough (e.g. such that $c n \in O(1)$) and $\std$ that satisfies the ``large'' noise condition, such that the optimal clairvoyant policy is not $\alpha$ times better than the policy that picks a box uniformly at random. The $\E[\D_{k:k}] \in \omega(\E[D])$ is crucial in this theorem, since, for the theorem to have bite, it must be that 
$\sqrt{\ln(cn)} \E[\D]$ is a lot smaller than $\E[\D_{n:n}]$ the reward of a prophet.

\section{Negative results for small noise environments} \label{sec:lower-bounds}

In this section, we show negative results for $\Naive$ (\Cref{subsec: naive fails}) and $\Linear_\gamma$ (\Cref{subsec: lower bound linear}). All missing proofs can be found in~\Cref{app: missing from 3}.

\subsection{Warm-up: Negative results for $\Naive$}\label{subsec: naive fails}


\begin{theorem}
\label{thm:master-lower-bound-s}
 For every distribution $\D$, all $n \geq 46$, and all $c \leq \frac{n-6\ln(n)}{n}$, there exists $\std^* = (\sigma^*_1, \dots, \sigma^*_n)$ such that $\std^* \in \good{\D}{n,c}$, and 
    \[ R_{\Naive}(\D, \std^*) \le \frac{8 \E[\D]}{\E[\D_{n:n}]} \cdot R_{\Opt_{\D}}(\D, \std^*)\]
\end{theorem}

As an immediate consequence of~\Cref{thm:master-lower-bound-s}, by picking a distribution $\D$ such that $\E[\D_{n:n}] \in \Theta( n \E[\D] )$, we get that $\Naive$ only gives a (trivial) $n$ approximation to the optimal policy.

\begin{corollary}
For all $n \geq 46$ and $c \leq \frac{n-6\ln(n)}{n}$, there exists $\D$ and $\std^* = (\sigma^*_1, \dots, \sigma^*_n)$ such that $R_{\Opt_{\D}}(\D, \std^*) \in \Omega( n ) \, R_{\Naive}(\D, \std^*)$.
\end{corollary}

\begin{proof}
Consider the distribution $\D$ that takes the value $0$ with probability $1-1/n$, and the value $n$ with probability $1/n$. Then, $\E[\D] = 1$, and $\E[ \D_{n:n} ] = n\cdot \left( 1 - \left( 1 - \frac{1}{n} \right)^n \right) \geq n \cdot \left( 1- \frac{1}{e} \right)$. Applying~\Cref{thm:master-lower-bound-s} implies the corollary.
\end{proof}

Our construction of $\std^*$ works as follows, where $c_b = 6 \ln n$ and $\sigma_b = 6 \beta^{(\D_{n:n})}_{n^2} \sqrt{\ln n}$:

\[\sigma^*_i = \begin{cases}
0 & i \in [1, n - c_b] \\
\sigma_b & i \in [n - c_b + 1, n] \\
\end{cases}\]

\noindent We refer to the boxes with $\sigma^*_i = 0$ as ``exact'', while the boxes with $\sigma^*_i = \sigma_b$ as having ``large noise.'' It is straightforward to confirm that $\std^* \in \good{\D}{n,c}$, for $c \leq \frac{n-6\ln(n)}{n}$ (according to~\Cref{dfn: small noise}).

\Cref{thm:master-lower-bound-s} will be an immediate consequence of two facts. First, intuitively, a large noise box will have large $\epsilon_i$ with high probability, and therefore be selected by $\Naive$, but its expected reward won't be much better than $4\E[\D]$ (\Cref{lem:naive-upper-s}). On the other hand, even the policy that selects the best exact box gets reward at least $\frac{1}{2} \E[\D_{n:n}]$ (\Cref{lem:optimal-lower-s}).

\begin{lemma}
\label{lem:optimal-lower-s}
For every distribution $\D$, for all $n \geq 46$, $R_{\Opt_\D}(\D, \std^*) \geq \frac{1}{2} \E[\D_{n:n}]$.
\end{lemma}

\begin{proof}
The optimal policy is as least as good as the policy that selects the box with the largest $y_i$ among the exact boxes. Since $x_i = y_i$ for these boxes, the reward of this policy is at least
\[
\E[\D_{n-c_b:n-c_b}] \ge^{\text{(\Cref{lem:compare-order-stats})}} \frac{n - c_b}{n} \cdot \E[\D_{n:n}] = \frac{n - 6 \ln n}{n} \cdot \E[\D_{n:n}] \ge^{(n \geq 46)} \frac{1}{2} \E[\D_{n:n}]. \qedhere
 \]
\end{proof}

\begin{lemma}
\label{lem:naive-upper-s}
    For every distribution $\D$, for all $n \geq 22$, we have that $R_{\Naive}(\D, \std^*) \le 4 \E[\D]$.
\end{lemma}

On a high level, our proof works as follows. Consider the event $\Event^*$ that $X_i \le \beta_{n^2}^{(\D_{n:n})}$ for all boxes $i$. We prove that conditioned on $\Event^*$, $\Naive$ gets an expected reward of at most $3 \E[\D]$. On the other hand, when $\Event^*$ does not occur, even if $\Naive$ performs as well as taking $\D_{n:n} = \max_i X_i$, the contribution to the final expected reward is also upper bounded by $\E[\D]$. The second fact can be shown directly from the definition of $\beta_{n^2}^{(\D_{n:n})}$. For the first fact, we first show that with high probability $\epsilon_i$ is not too small for some large box $i$ (\Cref{lem:large-noise-eps-bound-s}); conditioned on $\Event^*$ and this event, this implies that $\Naive$ picks a large noise box. It is also true that with high probability $\epsilon_i$ is not too big, for any large noise box $i$ (\Cref{lemma:large-box-reward}). Additionally conditioning on $\epsilon_i$ being not too big for every large noise box, we have that both the noise and the reward are not too big (and there is a box with large noise). We can then upper bound the reward of $\Naive$ by the reward of a ``clairvoyant'' policy which knows $\D$, but is required to pick a large noise box; for this step, we need a technical lemma (\Cref{lem:large-box-reward-bounded-D}) that will also be useful in our lower bound for linear policies. In all other events, we upper bound $\Naive$ by $\max_i X_i$.


\begin{lemma}
\label{lem:large-noise-eps-bound-s}
With probability at least $1 - \frac{1}{n^3}$, $\epsilon_i > \beta_{n^2}^{(\D_{n:n})}$ for at least one large noise box $i$.
\end{lemma}

\begin{lemma}
\label{lemma:large-box-reward}
For any large noise box $i$, we have $\Pr\left[ \epsilon_i \le 12 \beta_{n^2}^{(\D_{n:n})} \ln n \right] \geq 1 - \frac{1}{n^2}$.
\end{lemma}

\begin{lemma}
\label{lem:large-box-reward-bounded-D}
For any non-negative and bounded random variable $Z$ supported on $[0, V]$ and any $\sigma > 2V$, we have that $\E[Z \mid Z + \Norm(0,\sigma^2) = y] \le 2 \E[Z]$ for all $y \le \frac{\sigma^2}{2V}$.
\end{lemma}



\begin{proof}[Proof of~\Cref{lem:naive-upper-s}]

We define the following events.
\begin{itemize}[leftmargin=*]
\item $\Event_1$ be the event that $\epsilon_j \le 12 \beta_{n^2}^{(\D_{n:n})} \ln n$ for all large noise boxes $j$.
\item $\Event_1'$ be the event that $Y_j \le 18 \beta_{n^2}^{(\D_{n:n})} \ln n$ for all large noise boxes $j$.
\item $\Event_2$ be the event that $\epsilon_j > \beta_{n^2}^{(\D_{n:n})}$ for at least one large noise box $j$.
\item $\Event_2'$ be the event that $Y_j > \beta_{n^2}^{(\D_{n:n})}$ for at least one large noise box $j$.
\end{itemize}
Recall that $\Event^*$ is the event that $X_i \le \beta_{n^2}^{(\D_{n:n})}$ for all $i \in [n]$. 

We first explore the relationship between these events. First, 
notice that if $X_i \leq \beta_{n^2}^{(\D_{n:n})}$ and $\epsilon_i \leq 12 \beta_{n^2}^{(\D_{n:n})} \ln n$, we have that 
\[Y_i = X_i + \epsilon_i \le \beta_{n^2}^{(\D_{n:n})} + 12 \beta_{n^2}^{(\D_{n:n})} \ln n \le 18 \beta_{n^2}^{(\D_{n:n})} \ln n.\]
Therefore, $\Event_1 \cap \Event^* \subseteq \Event_1' \cap \Event^*$. Since $X_i \geq 0$ for all $i$, $\Event_2'$ occurs every time $\Event_2$ occurs, i.e. $\Event_2 \subseteq \Event_2'$, and thus $\Event_2 \cap \Event^* \subseteq \Event_2' \cap \Event^*$. Therefore, $\Event_1 \cap \Event_2 \cap \Event^* \subseteq \Event_1' \cap \Event_2' \cap \Event^*$, or $\overline{\Event_1 \cap \Event_2} \cap \Event^* \supseteq \overline{\Event_1' \cap \Event_2'} \cap \Event^*$. 

First, we will bound $\E[\max X_i \mid \overline{\Event_1' \cap \Event_2'} \cap \Event^*] \cdot \Pr[\overline{\Event_1' \cap \Event_2'} \mid \Event^*]$, which is an upper bound on the contribution of outcomes in $\overline{\Event_1' \cap \Event_2'} \cap \Event^*$ to the overall expected reward of $\Naive$. Since the contribution of an event $A$ to the expectation of a random variable ($\E[X|A]\Pr[A]$) is smaller than the contribution of an event $B$ to the expectation if $A \subseteq B$, we have
\[\E[\max_i X_i \mid \overline{\Event_1' \cap \Event_2'} \cap \Event^*] \cdot \Pr[\overline{\Event_1' \cap \Event_2'} \mid \Event^*] \leq \E[\max_i X_i \mid \overline{\Event_1 \cap \Event_2} \cap \Event^*] \cdot \Pr[\overline{\Event_1 \cap \Event_2} \mid \Event^*].\]
By~\Cref{lemma:large-box-reward}, $\Pr[\Event_1] \geq \left( 1 - \frac{1}{n^2} \right)^{c_b} \ge 1 - \frac{6 \ln n}{n^2}$. By~\Cref{lem:large-noise-eps-bound-s}, $\Pr[\Event_2] \geq 1 - \frac{1}{n^3}$. Therefore, $\Pr[ \Event_1 \cap \Event_2 ] \geq \Pr[ \Event_1] + \Pr[\Event_2] - 1 \geq 1 - \frac{6 \ln n}{n^2} + 1 - \frac{1}{n^3} - 1 \geq 1 - \frac{7 \ln n}{n^2}$. Observe that, $\Event_1$ and $\Event_2$ are independent of the $X_i$s, while $\Event^*$ only dependent on $X_i$s. Therefore, $\Event_1 \cap \Event_2$ and $\Event^*$ are independent, and hence $\Pr[\Event_1 \cap \Event_2 \mid \Event^*] = \Pr[\Event_1 \cap \Event_2] \ge 1 - \frac{7 \ln n}{n^2}$, or $\Pr[\overline{\Event_1 \cap \Event_2} \mid \Event^*] \le \frac{7 \ln n}{n^2}$. Additionally, $\E[\max_i X_i \mid \overline{\Event_1 \cap \Event_2} \cap \Event^*] = \E[\max_i X_i \mid \Event^*]$, as $\Event_1$ and $\Event_2$ are events regarding $\epsilon_i$s and therefore is independent of $X_i$. Furthermore, $\E[\max_i X_i \mid \Event^*] = \E[\max_i X_i \mid X_i \leq \beta_{n^2}^{(\D_{n:n})}] \le \E[\max_i X_i] = \E[\D_{n:n}]$. Putting everything together, we have
\begin{align*}
\E[\max_i X_i \mid \overline{\Event_1' \cap \Event_2'} \cap \Event^*] \cdot \Pr[\overline{\Event_1' \cap \Event_2'} \mid \Event^*] &\le \E[\max_i X_i \mid \overline{\Event_1 \cap \Event_2} \cap \Event^*] \cdot \Pr[\overline{\Event_1 \cap \Event_2} \mid \Event^*] \\
&\le \E[\D_{n:n}] \cdot \frac{7 \ln n}{n^2} \\
&\le^{\text{(\Cref{lem:compare-order-stats})}} \frac{7 \ln n}{n^2} \cdot n \cdot \E[\D] \\
&\le^{(n \geq 22)} \E[\D].
\end{align*}

Second, we will upper bound the contribution of outcomes in $\Event_1' \cap \Event_2' \cap \Event^*$ to the expected reward of $\Naive$. Note that in such outcomes, $\Naive$ must choose a large noise box, by the definition of $\Event_2'$ ($Y_j > \beta_{n^2}^{(\D_{n:n})}$ for some large noise box $j$) and $\Event^*$ ($X_i \leq \beta_{n^2}^{(\D_{n:n})}$ for all $i$, and therefore the exact boxes). Therefore, in such an outcome, the reward of $\Naive$ is at most the reward of an optimal policy which also knows $\D$, but is conditioned to pick a large noise box. When selecting box $i$ such a policy makes expected reward $\E[X_i \mid Y_i = y_i, \Event^*, \Event_1', \Event_2'] = \E[X_i \mid Y_i = y_i, \Event^*]$, where the equality holds since $X_i$ is independent of $Y_j$, for $j \neq i$, and $\Event_1' \cap \Event_2'$ have less information about $Y_i$ than $\{ Y_i = y_i \}$. Let $R_i(y_i) = \E[ X_i \mid Y_i = y_i, \Event^* ]$. The reward of an optimal policy which knows $\D$ and is conditioned to pick a large noise box is then $\E_{\observ}\left[ \max_{i \in [n - c_b+1,n]} R_i(y_i) \mid \Event_1' \cap \Event_2' \cap \Event^* \right]$. We prove that $R_i(y_i) \leq 2 \E[\D]$ for all $y_i$ consistent with $\Event_1' \cap \Event_2' \cap \Event^*$, which in turn implies an upper bound of $2\E[\D]$ for the expected reward of $\Naive$ conditioned on in $\Event_1' \cap \Event_2' \cap \Event^*$.

Consider any large noise box $i$. Let $\overline{X}_i = X_i \mid X_i \leq \beta_{n^2}^{(\D_{n:n})}$.\footnote{Equivalently, we can think of sampling from $\overline{X}_i$ by sampling from $X_i$, until $X_i \leq \beta_{n^2}^{(\D_{n:n})}$.} Then, conditioned on $\Event_1' \cap \Event_2' \cap \Event^*$, for any realization of $\observ$, we note that $R_i(y_i) = \E[X_i \mid Y_i = y_i, \Event^*] = \E[\overline{X}_i \mid \overline{X}_i + \Norm(0, \sigma_i^2) = y_i]$. Furthermore, as $y_i$ is a realization conditioned on $\Event_1' \cap \Event_2' \cap \Event^*$, we have $y_i \le 18 \beta_{n^2}^{(\D_{n:n})} \ln n$.
Using~\Cref{lem:large-box-reward-bounded-D} for $V = \beta_{n^2}^{(\D_{n:n})}$ and $\sigma = \sigma_b = 6 \beta_{n^2}^{(\D_{n:n})} \sqrt{\ln n}$, we have $\E[\overline{X}_i \mid \overline{X}_i + \Norm(0, \sigma_i^2) = y_i] \le 2 \E[\overline{X}_i] \le 2 \E[X_i] = 2 \E[\D]$.

Overall, conditioned on $\Event^*$, if $\Event_1' \cap \Event_2'$ occurs, $\Naive$'s expected reward is at most $2 \E[\D]$; otherwise, the contribution to the expected reward is at most $\E[\D]$. Thus, the reward of $\Naive$ conditioned on $\Event^*$ is at most
\begin{align*}
\Pr[\Event_1' \cap \Event_2' \mid \Event^*] \cdot 2 \E[\D] &+ \E[\max X_i \mid \overline{\Event_1' \cap \Event_2'} \cap \Event^*] \cdot \Pr[\overline{\Event_1' \cap \Event_2'} \mid \Event^*] \\
        &\le 2 \E[\D] + \E[\D] \\
        &= 3 \E[\D].
\end{align*}
Finally, conditioned on $\Event^*$ not happening, the best $\Naive$ can do is $\D_{n:n} = \max_i X_i$, whose expected reward is $\E[\D_{n:n} \mid \overline{\Event^*}]$. Therefore:
\begin{align*}
    R_{\Naive}(\D, \std^*) &\leq 3 \E[\D] \cdot \Pr[\Event^*] + \E[\D_{n:n} \mid \overline{\Event^*}] \cdot \Pr[\overline{\Event^*}] \\
    &\le 3 \E[\D] + \E[\D_{n:n} \mid \D_{n:n} \ge \beta_{n^2}^{(\D_{n:n})}] \cdot \Pr[\D_{n:n} \ge \beta_{n^2}^{(\D_{n:n})}] \\
    &\le^{\text{(\Cref{dfn: beta_m})}} 3 \E[\D] + \frac{\E[ \D_{n:n} ]}{n^2} \\
    &\le^{\text{(\Cref{lem:compare-order-stats})}} 3 \E[\D] + \frac{n \cdot \E[\D]}{n^2} \\
    &\le 4 \E[\D]. \qedhere
\end{align*}

\end{proof}
\begin{proof}[Proof of~\Cref{thm:master-lower-bound-s}]
The theorem is implied by Lemmas~\ref{lem:optimal-lower-s} and~\ref{lem:naive-upper-s}.
\end{proof}

\subsection{Negative results for Linear policies}\label{subsec: lower bound linear}

In this section, we give our negative results for linear policies. Recall that a linear policy parameterized $\gamma: \R^n \times \R^n \to \R$ selects the box which maximizes $y_i - \gamma(\std, \observ) \cdot \sigma_i$.

\begin{theorem}
\label{thr:master-lower-bound}
For every MHR distribution $\D$, for all $n \geq n_0$, for some constant $n_0$, there exists $\std^* = (\sigma^*_1,\dots,\sigma^*_n)$, such that $\std^* \in \goodMHR{\D}{n,1/5626}$, and for every function $\gamma : \R^n \times \R^n \to \R$, we have
    \[
    R_{\Linear_\gamma}(\D, \std^*) \in O \left( \frac{\E[\D]}{\E[\D_{n:n}]} \right) \, R_{\Opt_{\D}}(\D, \std^*) .
    \]
\end{theorem}

An immediate corollary is that linear policies give, in the worst case, a logarithmic approximation, even for MHR distributions, by considering $\D$ to be the exponential distribution with parameter $\lambda=1$, for which $\E[\D_{n:n}] = \sum_{i=1}^n \frac{1}{i} \geq \ln n$. Note also that $\E[\D_{n:n}] \leq \ln n + 1$ for all MHR random variables (\Cref{lem:order-stat-mean-bound}; \Cref{subsec: lower bound linear}), so the exponential distribution minimizes the ratio in~\Cref{thr:master-lower-bound} (up to constants). 

\begin{corollary}
There exists $\D$, such that for all $n \geq n_0$, for some constant $n_0$, there exists $\std^* \in \goodMHR{\D}{n,1/5626}$ such that $R_{\Opt_{\D}}(\D, \std^*) \in \Omega( \ln(n) ) \cdot R_{\Linear_\gamma}(\D, \std^*)$.
\end{corollary}

Our construction of $\std^*$ works as follows. It contains one box such that $\sigma^* = 0$, a small number of boxes with some small noise $\sigma_s$, and the remaining boxes have large noise $\sigma_b$:
\[
\sigma^*_i = \begin{cases}
0 & i = 1 \\
\sigma_s & i \in [2, c_s + 1] \\
\sigma_b & i \in [c_s + 2, n] \\
\end{cases}
\]
where $c_s = n^{1/5626}$, $\sigma_s = \frac{37}{9\sqrt{2}} \frac{\E[\D_{c_s:c_s}]}{\sqrt{\ln n}}$, and $\sigma_b = 6 \alpha_{n^{1/10000}}^{(\D_{n-c_s:n-c_s})} \sqrt{\ln n}$.
We refer to the first box as the ``exact box,'' the boxes with $\sigma^*_i = \sigma_s$ as ``small noise'' boxes, and the rest as ``large noise'' boxes. One can easily confirm that $\std^* \in \goodMHR{\D}{n,1/5626}$.



We first lower bound the expected reward of the optimal policy.

\begin{lemma}
\label{lem:optimal-lower}
For every MHR distribution $\D$, all $n \geq n_0$, for some constant $n_0$, $R_{\Opt}(\D, \std^*) \in \Omega \left( \E[\D_{n:n}] \right)$.
\end{lemma}

\begin{proof}
The optimal policy is at least as good as the policy that picks the box with the largest $y_i$ among the small noise boxes. Consider the event that $|\epsilon_i| \le \frac{\sqrt{2}}{75} \sigma_s \sqrt{\ln n}$ for all small noise boxes $i$:

\begin{align*}
    \Pr\left[ \max_{i \in [2, c_s + 1]} |\epsilon_i| \le \frac{\sqrt{2}}{75} \sigma_s \sqrt{\ln n} \right] &= \Pr\left[ |\Norm(0, \sigma_s^2)| \le \frac{\sqrt{2}}{75} \sigma_s \sqrt{\ln n} \right]^{c_s} \\
    &= \left(2 \Phi \left( \frac{\sqrt{2 \ln n}}{75} \right) - 1\right)^{c_s} \\
    &\ge^{\text{(\Cref{lem:normal-bound})}} \left(2  \left( 1 - \frac{1}{\sqrt{2 \pi}} \frac{75}{\sqrt{2 \ln n}} \exp\left( -\frac{1}{2} \cdot \frac{2}{5625} \ln n \right) \right) - 1\right)^{c_s} \\
    &= \left(1 - \frac{75}{\sqrt{\pi \ln n}} n^{-1/5625}\right)^{n^{1/5626}} \\
    &\ge^{\text{(Bernoulli's inequality)}} 1 - \frac{75}{\sqrt{\pi \ln n}} n^{1/5626 - 1/5625} \\
    &\ge
    1 - \frac{1}{\ln n}
\end{align*}

When this event occurs, the reward from picking a small noise box $i$ is at least $y_i - \frac{\sqrt{2}}{75} \sqrt{\ln n}\sigma_s$, and therefore the overall reward of picking from small noise boxes is at least $\max_{i=2,\dots,c_s+1} x_i - \frac{2\sqrt{2}}{75} \sqrt{\ln n}\sigma_s$. Noting that the noise and reward are independent random variables, we have:
\begin{align*}
    R_{\Opt}(\D, \std^*) &\ge \left( 1 - \frac{1}{\ln n} \right) \cdot \left( \E \left[ \max_{i \in [2, c_s + 1]} X_i \right] - \frac{2 \sqrt{2}}{75} \sqrt{\ln n} \cdot \sigma_s \right) \\
    &= \left( 1 - \frac{1}{\ln n} \right)\left( \E[\D_{c_s:c_s}] - \frac{2 \sqrt{2}}{75} \sqrt{\ln n} \cdot \frac{5}{\sqrt{2}} \frac{\E[\D_{c_s:c_s}]}{\sqrt{\ln n}} \right) \\
    &\ge \left( 1 - \frac{1}{\ln n} \right) \cdot \frac{1}{75} \E[\D_{c_s:c_s}]. \\
\end{align*}

\noindent The following lemma allows us to bound $\E[\D_{c_s:c_s}]$ as a function of $\E[\D_{n:n}] $:
\begin{lemma}
\label{lem:order-stat-order-stat-bound}
    For any MHR distribution $\D$ supported on $[0, \infty)$, for any $n \ge 4$ and $a \ge 1$, we have
    \[\E[\D_{n^a:n^a}] \le 4a \cdot \E[\D_{n:n}].\]
\end{lemma}

Continuing our derivation
\[
R_{\Opt_{\D}}(\D, \std^*) \ge^{\text{(\Cref{lem:order-stat-order-stat-bound})}}
\left( 1 - \frac{1}{\ln n} \right) \frac{1}{75} \cdot \frac{1}{4 \cdot 5626}  \E[\D_{n:n}]\geq^{(n\geq e^{606})} 
\frac{1}{2\,000\,000} \E[\D_{n:n}]. \qedhere
\]

\end{proof}

Our next (and final) task is to upper bound the expected reward of $\Linear$. The main lemma for this stage is as follows.

\begin{lemma}
\label{lem:linear-upper}
For every MHR distribution $\D$, for all $n \geq n_0$, for some constant $n_0$, and for all $\gamma$, it holds that $R_{\Linear_\gamma}(\D, \std^*) \le 8 \E[\D]$.
\end{lemma}

The proof structure is similar to~\Cref{lem:naive-upper-s}. We first prove (\Cref{lem:linear-upper-regime}) that conditioned on an event $\Event^*$, $\Linear_\gamma$'s expected reward is upper bounded, while the contribution to the reward of other events is negligible, even if $\Linear_\gamma$ performs as well as taking $\max_i X_i$. Here, $\Event^*$ is the event that \emph{$X_i \le \alpha_{n^{1/10000}}^{(\D_{c_s:c_s})}$ for all small noise boxes $i$, and $X_j \le \alpha_{n^{1/10000}}^{(\D_{n - c_s:n - c_s})}$ for all remaining boxes $j$}.

\begin{lemma}
\label{lem:linear-upper-regime}
For every MHR distribution $\D$, for all $n \geq n_0$, for some constant $n_0$, and for all $\gamma$, the expected reward of a policy $\Linear_\gamma$ conditioned on the event $\Event^*$ is at most $7 \E[\D]$.  
\end{lemma}

To prove~\Cref{lem:linear-upper-regime}, we first consider a slightly different family of policies. Let $\LinearFixed_c$ be the policy that chooses the box with the largest $y_i - c \sigma_i$, where $c$ is a constant independent of $\observ$ and $\std$. We show that with high probability, all $\LinearFixed$ policies make poor choices. We can use this fact to get bounds on the performance of $\Linear_\gamma$ (conditioned on certain events), since, fixing $\observ$ and $\std$, $\Linear_\gamma$ is only as good as the best $\LinearFixed$ policy. We consider two cases on $c$: $c > \theta^*$ and $c \le \theta^*$, where $\theta^* = \sqrt{\frac{\ln n}{2}}$.

To make the presentation cleaner, we define the following events.

\begin{definition}
\label{dfn:events}
    Let
\begin{itemize}
    \item $\Event_1$ be the event of $\max_{i \in [2, c_s + 1]} \epsilon_i \le \frac{\theta^* \sigma_s}{37}$.
    \item $\Event_2$ be the event of $\max_{i \in [c_s + 2, n]} \epsilon_i - \theta^* \sigma_b \ge \sigma_b$.
    \item $\Event_2'$ be the event of $\max_{i \in [c_s + 2, n]} Y_i - c \sigma_b \ge \sigma_b$ for all $c < \theta^*$.
    \item $\Event_3$ be the event of $\max_{i \in [c_s + 2, n]} \epsilon_i \le 12 \alpha_{n^{1/10000}}^{(\D_{n-c_s:n-c_s})} \ln n$.
    \item $\Event_3'$ be the event of $\max_{i \in [c_s + 2, n]} Y_i \le 18 \alpha_{n^{1/10000}}^{(\D_{n-c_s:n-c_s})} \ln n$.
\end{itemize}
Recall that $\Event^*$ is the event that $X_i \le \alpha_{n^{1/10000}}^{(\D_{c_s:c_s})}$ for all small noise boxes $i$, and $X_j \le \alpha_{n^{1/10000}}^{(\D_{n - c_s:n - c_s})}$ for all remaining boxes $j$.
\end{definition}

We state some technical lemmata. \Cref{lem:c-large} and~\Cref{lem:c-small} say that if various combinations of the above events occur, $\LinearFixed_c$ policies make bad choices.

\begin{lemma}\label{lem:eps-small-bound}
For all $n \geq n_0$, for some constant $n_0$, $\Pr[\Event_1] \geq 1 - \frac{1}{\ln n}$.
\end{lemma}

\begin{lemma}\label{lem:eps-large-bound}
For all $n \geq n_0$, for some constant $n_0$, $\Pr[\Event_2] \geq 1 - \frac{1}{\ln n}$.
\end{lemma}

\begin{lemma}
\label{lem:large-box-reward-linear}
For all $n \geq n_0$, for some constant $n_0$, for any large noise box $i$, \[ \Pr\left[ Y_i \le 18 \alpha_{n^{1/10000}}^{(\D_{n-c_s:n-c_s})} \ln n \right] \geq 1 - \frac{1}{n^2}.\]
\end{lemma}

\begin{lemma}
\label{lem:c-large}
If $\Event^* \cap \Event_1$ occurs, for all $c \geq \theta^*$, $\LinearFixed_c$ does not choose a small noise box.
\end{lemma}

\begin{lemma}
\label{lem:c-small}
If $\Event^* \cap \Event_1 \cap \Event_2'$ occurs, for all $c < \theta^*$, $\LinearFixed_c$ chooses some large noise box.
\end{lemma}




We can now prove~\Cref{lem:linear-upper-regime}:

\begin{proof}[Proof of~\Cref{lem:linear-upper-regime}]
We first explore the relationship between the events defined in~\Cref{dfn:events}. First, note that $\Event_2 \subseteq \Event_2'$: if $\max_{i \in [c_s + 2, n]} \epsilon_i - \theta^* \sigma_b \ge \sigma_b$, then for all $c < \theta^*$ we have
\[\max_{i \in [c_s + 2, n]} Y_i - c \sigma_b = \max_{i \in [c_s + 2, n]} (X_i + \epsilon_i) - c \sigma_b \ge \max_{i \in [c_s + 2, n]} \epsilon_i - \theta^* \sigma_b \ge \sigma_b.\]
Second, note that $\Event^* \cap \Event_3 \subseteq \Event_3'$, or $\Event^* \cap \Event_3 \subseteq \Event^* \cap \Event_3'$: if $\max_{i \in [c_s + 2, n]} X_i \le \alpha_{n^{1/10000}}^{(\D_{n - c_s:n - c_s})}$ and $\max_{i \in [c_s + 2, n]} \epsilon_i \le 12 \alpha_{n^{1/10000}}^{(\D_{n-c_s:n-c_s})} \ln n$, then
\begin{align*}
\max_{i \in [c_s + 2, n]} Y_i &= \max_{i \in [c_s + 2, n]} X_i + \epsilon_i \\
    &\le \max_{i \in [c_s + 2, n]} X_i + \max_{i \in [c_s + 2, n]} \epsilon_i \\
    &\le \alpha_{n^{1/10000}}^{(\D_{n - c_s:n - c_s})} + 12 \alpha_{n^{1/10000}}^{(\D_{n-c_s:n-c_s})} \ln n \\
    &\le 18 \alpha_{n^{1/10000}}^{(\D_{n-c_s:n-c_s})}.
\end{align*}
Ultimately, we have $\Event_1 \cap \Event_2 \cap \Event_3 \cap \Event^* \subseteq \Event_1 \cap \Event_2' \cap \Event_3' \cap \Event^*$, or $\overline{\Event_1 \cap \Event_2 \cap \Event_3} \cap \Event^* \supseteq \overline{\Event_1 \cap \Event_2' \cap \Event_3'} \cap \Event^*$.

We now bound $\E[\max_i X_i \mid \overline{\Event_1 \cap \Event_2' \cap \Event_3'} \cap \Event^*] \cdot \Pr[\overline{\Event_1 \cap \Event_2' \cap \Event_3'} \mid \Event^*]$, which is an upper bound on the contribution of outcomes in $\overline{\Event_1 \cap \Event_2' \cap \Event_3'} \cap \Event^*$ to the overall expected reward of $\Linear_\gamma$.
\begin{align*}
\E[\max_i X_i \mid \overline{\Event_1 \cap \Event_2' \cap \Event_3'} \cap \Event^*] &\cdot \Pr[\overline{\Event_1 \cap \Event_2' \cap \Event_3'} \mid \Event^*] \\
    &\le \E[\max_i X_i \mid \overline{\Event_1 \cap \Event_2 \cap \Event_3} \cap \Event^*] \cdot \Pr[\overline{\Event_1 \cap \Event_2 \cap \Event_3} \mid \Event^*]
\end{align*}
By~\Cref{lem:eps-small-bound}, $\Pr[\Event_1] \ge 1 - \frac{1}{\ln n}$. By~\Cref{lem:eps-large-bound}, $\Pr[\Event_2] \ge 1 - \frac{1}{\ln n}$. Using~\Cref{lem:large-box-reward-linear}, $\Pr[\Event_3] \ge (1 - \frac{1}{n^2})^{n - c_s - 1} \ge 1 - \frac{1}{n}$. Therefore, by the a union bound, $\Pr[\Event_1 \cap \Event_2 \cap \Event_3] \ge 1 - \frac{2}{\ln n} - \frac{1}{n} \ge 1 - \frac{3}{\ln n}$. Observe that, $\Event_1$, $\Event_2$, and $\Event_3$ are independent of the $X_i$s, while $\Event^*$ only dependent on $X_i$s. Therefore, $\Event_1 \cap \Event_2 \cap \Event_3$ and $\Event^*$ are independent, and hence $\Pr[\Event_1 \cap \Event_2 \cap \Event_3 \mid \Event^*] = \Pr[\Event_1 \cap \Event_2 \cap \Event_3] \ge 1 - \frac{3}{\ln n}$, or $\Pr[\overline{\Event_1 \cap \Event_2 \cap \Event_3} \mid \Event^*] \le \frac{3}{\ln n}$. Additionally, $\E[\max_i X_i \mid \overline{\Event_1 \cap \Event_2 \cap \Event_3} \cap \Event^*] = \E[\max_i X_i \mid \Event^*]$, as $\Event_1$, $\Event_2$, and $\Event_3$ are events regarding $\epsilon_i$s and therefore independent of $X_i$. Finally, $\E[\max_i X_i \mid \Event^*] \le \E[\max_i X_i] = \E[\D_{n:n}]$, as $\Event^*$ is an event which upper bounds $X_i$. Putting everything together:
\begin{align*}
\E[\max_i X_i \mid \overline{\Event_1 \cap \Event_2' \cap \Event_3'} \cap \Event^*] &\cdot \Pr[\overline{\Event_1 \cap \Event_2' \cap \Event_3'} \mid \Event^*] \\
    &\le \E[\max_i X_i \mid \overline{\Event_1 \cap \Event_2 \cap \Event_3} \cap \Event^*] \cdot \Pr[\overline{\Event_1 \cap \Event_2 \cap \Event_3} \mid \Event^*] \\
    &\le \frac{3}{\ln n} \cdot \E[\D_{n:n}] \\
    &\le^{\text{(\Cref{lem:order-stat-mean-bound})}} \frac{3}{\ln n} \cdot (\ln n + 1) \E[\D] \\
    &\le 4 \E[\D].
\end{align*}
Next, we will upper bound the contribution of outcomes in $\Event_1 \cap \Event_2' \cap \Event_3' \cap \Event^*$ to the expected reward of $\Linear_\gamma$. Note that in such outcomes, for every $c_1 \ge \theta^*$ and $c_2 < \theta^*$, $\LinearFixed_{c_1}$ does not choose a small noise box (\Cref{lem:c-large}) and $\LinearFixed_{c_2}$ chooses some large noise box (\Cref{lem:c-small}). Hence, in such outcomes, $\Linear_\gamma$ does not choose a small noise box.
Therefore, in such an outcome, the reward of $\Linear_\gamma$ is at most the reward of an optimal policy that knows $\D$, but is conditioned to not pick a small noise box. When selecting box $i$, such a policy has expected reward $\E[X_i \mid Y_i = y_i, \Event^*, \Event_1, \Event_2', \Event_3']$. We first observe that $\E[X_i \mid Y_i = y_i, \Event^*, \Event_1, \Event_2', \Event_3'] = \E[X_i \mid Y_i = y_i, \Event_2', \Event_3']$ as $\Event_1$ regards $\epsilon_j$ of all small noise boxes $j$, which are never picked in this policy. Secondly, $\E[X_i \mid Y_i = y_i, \Event_2', \Event_3'] = \E[X_i \mid Y_i = y_i, \Event^*]$ as $X_i$ is independent of $Y_j$, for $j \neq i$, and $\Event_2' \cap \Event_3'$ have less information about $Y_i$ than $\{ Y_i = y_i \}$.

Let $R_i(y_i) = \E[ X_i \mid Y_i = y_i, \Event^* ]$. The reward of an optimal policy which knows $\D$ and is conditioned to not pick a small noise box is then
\begin{align*}
    \E_{\observ}&\left[ \max_{i \in \{1\} \cup [n - c_b+1,n]} R_i(y_i) \mid \Event_1 \cap \Event_2' \cap \Event_3' \cap \Event^* \right] \\
    &\le^{(R_1(y_1) \ge 0)} \E_{\observ}\left[ R_1(y_1) + \max_{i \in [n - c_b+1,n]} R_i(y_i) \mid \Event_1 \cap \Event_2' \cap \Event_3' \cap \Event^* \right] \\
    &=^{(\sigma_1 = 0)} \E[ X_1 \mid \Event_1 \cap \Event_2' \cap \Event_3' \cap \Event^* ] + \E_{\observ}\left[\max_{i \in [n - c_b+1,n]} R_i(y_i) \mid \Event_1 \cap \Event_2' \cap \Event_3' \cap \Event^* \right] \\
    &= \E[X_1 \mid \Event^*] + \E_{\observ}\left[\max_{i \in [n - c_b+1,n]} R_i(y_i) \mid \Event_1 \cap \Event_2' \cap \Event_3' \cap \Event^* \right],
\end{align*}
where the last inequality holds since $\Event_1$, $\Event_2$, and $\Event_3$ are events regarding small noise and large noise boxes, and hence is independent of $X_1$.

Consider any small noise box $i$. Let $\overline{X}_i = X_i \mid X_i \leq \alpha_{n^{1/10000}}^{(\D_{n - c_s:n - c_s})}$. 
Then, conditioned on $\Event_1 \cap \Event_2' \cap \Event_3' \cap \Event^*$, for any realization of $\observ$, we note that $R_i(y_i) = \E[X_i \mid Y_i = y_i, \Event^*] = \E[\overline{X}_i \mid \overline{X}_i + \Norm(0, \sigma_i^2) = y_i]$. Furthermore, as $y_i$ is a realization conditioned on $\Event_1 \cap \Event_2' \cap \Event_3' \Event^*$, we have $y_i \le 18 \alpha_{n^{1/10000}}^{(\D_{n - c_s:n - c_s})} \ln n$.
Using~\Cref{lem:large-box-reward-bounded-D} with $V = \beta_{n^2}^{(\D_{n:n})}$ and $\sigma = \sigma_b = 6\alpha_{n^{1/10000}}^{(\D_{n - c_s:n - c_s})} \sqrt{\ln n}$, we have $\E[\overline{X}_i \mid \overline{X}_i + \Norm(0, \sigma_i^2) = y_i] \le 2 \E[\overline{X}_i] \le 2 \E[X_i] = 2 \E[\D]$. As this is true for any small noise box $i$ on any realization of $\observ$, we then have
\begin{align*}
    \E_{\observ}&\left[ \max_{i \in \{1\} \cup [n - c_b+1,n]} R_i(y_i) \mid \Event_1 \cap \Event_2' \cap \Event_3' \cap \Event^* \right] \\
    &\le \E[X_1 \mid \Event^*] + \E_{\observ}\left[\max_{i \in [n - c_b+1,n]} R_i(y_i) \mid \Event_1 \cap \Event_2' \cap \Event_3' \cap \Event^* \right] \\
    &\le \E[X_1 \mid X_1 \le \alpha_{n^{1/10000}}^{(\D_{n - c_s:n - c_s})}] + \E_{\observ}[2 \E[\D]] \\
    &\le \E[X_1] + 2 \E[\D] \\
    &= 3\E[\D].
\end{align*}

Overall, conditioned on $\Event^*$, if $\Event_1 \cap \Event_2' \cap \Event_3'$ occurs, $\Naive$'s expected reward is at most $3 \E[\D]$, while otherwise, the contribution to the expected reward is at most $4\E[\D]$. Therefore, the reward of $\Naive$ conditioned on $\Event^*$ is at most $7 \E[\D]$.
\end{proof}

With~\Cref{lem:linear-upper-regime} at hand, we can prove~\Cref{lem:linear-upper}.

\begin{proof}[Proof of~\Cref{lem:linear-upper}]
We decompose $\Event^*$ as $\Event^*_1 \cap \Event^*_2$, where $\Event^*_1$ and $\Event^*_2$ are two independent events defined as follows. $\Event^*_1$ is the event that $X_i \le \alpha_{n^{1/10000}}^{(\D_{c_s:c_s})}$ for all small noise boxes $i \in [2, c_s+1]$. $\Event^*_2$ is the event that $X_j \le \alpha_{n^{1/10000}}^{(\D_{n - c_s:n - c_s})}$ for all remaining boxes $j$.

Observe that $\Pr[\overline{\Event^*_1}] = \Pr\left[\max_{i \in [2, c_s + 1]} X_i > \alpha_{n^{1/10000}}^{(\D_{c_s:c_s})}\right] = \Pr[\D_{c_s:c_s} > \alpha_{n^{1/10000}}^{(\D_{c_s:c_s})}] = \frac{1}{n^{1/10000}}$. Similarly, $\Pr[\overline{\Event^*_2}] = \frac{1}{n^{1/10000}}$. Therefore, $\Pr[\overline{\Event^*}] = \Pr[\overline{\Event^*_1} \cup \overline{\Event^*_2}] \le \Pr[\overline{\Event^*_1}] + \Pr[\overline{\Event^*_2}] = \frac{2}{n^{1/10000}}$.

Next, we upper bound the contribution of $\overline{\Event^*}$ to the overall reward of $\Linear_\gamma$. Overloading notation, let $R_{\Linear_\gamma}(\D,\std^* \mid \overline{\Event^*})$ be the expected reward of $\Linear_\gamma$ when $\overline{\Event^*}$occurs. Then, we have
\begin{align*}
    &R_{\Linear_\gamma}(\D,\std^* \mid \overline{\Event^*}) \cdot \Pr[\overline{\Event^*}] \le \E[\max_i X_i \mid \overline{\Event^*_1} \cup \overline{\Event^*_2}] \cdot \Pr[\overline{\Event^*_1} \cup \overline{\Event^*_2}] \\
    &\quad \le \left(\E\left[\max_{i \in [2, c_s + 1]} X_i \mid \overline{\Event^*_1} \cup \overline{\Event^*_2}\right] + \E\left[\max_{i \in [1, n] \setminus [2, c_s + 1]} X_i \mid \overline{\Event^*_1} \cup \overline{\Event^*_2}\right]\right) \cdot \Pr[\overline{\Event^*_1} \cup \overline{\Event^*_2}] \\
    &\quad = \left(\E\left[\max_{i \in [2, c_s + 1]} X_i \mid \overline{\Event^*_1}\right] + \E\left[\max_{i \in [1, n] \setminus [2, c_s + 1]} X_i \mid \overline{\Event^*_2}\right]\right) \cdot \Pr[\overline{\Event^*_1} \cup \overline{\Event^*_2}] \\
    &\quad = \left(\E\left[\D_{c_s:c_s} \mid \D_{c_s:c_s} > \alpha_{n^{1/10000}}^{(\D_{c_s:c_s})}\right] + \E\left[\D_{n - c_s:n - c_s} \mid \D_{n - c_s:n - c_s} > \alpha_{n^{1/10000}}^{(\D_{n - c_s:n - c_s})}\right]\right) \cdot \Pr[\overline{\Event^*_1} \cup \overline{\Event^*_2}] \\
    &\quad \leq 2 \left(\frac{\E\left[\D_{c_s:c_s} \mid \D_{c_s:c_s} > \alpha_{n^{1/10000}}^{(\D_{c_s:c_s})}\right]}{n^{1/10000}} + \frac{\E\left[\D_{n - c_s:n - c_s} \mid \D_{n - c_s:n - c_s} > \alpha_{n^{1/10000}}^{(\D_{n - c_s:n - c_s})}\right]}{n^{1/10000}}\right).
\end{align*}
Note that $\frac{\E\left[\D_{c_s:c_s} \mid \D_{c_s:c_s} > \alpha_{n^{1/10000}}^{(\D_{c_s:c_s})}\right]}{n^{1/10000}} = \E\left[\D_{c_s:c_s} \mid \D_{c_s:c_s} > \alpha_{n^{1/10000}}^{(\D_{c_s:c_s})}\right] \cdot \Pr\left[\D_{c_s:c_s} > \alpha_{n^{1/10000}}^{(\D_{c_s:c_s})}\right]$, and similarly for the second term. In the appendix, we show, stated as~\Cref{lem:very specific lemma on 6 log n}, that for every MHR distribution $\D$, $n \geq 1$ and $m \geq 2$:
\[\E[\D_{n:n} \mid \D_{n:n} > \alpha_m^{(\D_{n:n})}] \cdot \Pr[\D_{n:n} > \alpha_m^{(\D_{n:n})}] \le \frac{15 (\ln m + \ln n + 1)\E[\D]}{2m}. \]

Applied here (noting that $\D_{a:a}$ is MHR for all $a \geq 1$; see~\Cref{lem:mhr-order-stat}), we have:
\[\frac{\E\left[\D_{c_s:c_s} \mid \D_{c_s:c_s} > \alpha_{n^{1/10000}}^{(\D_{c_s:c_s})}\right]}{n^{1/10000}} \le \frac{15 (\ln(n^{1/10000}) + \ln(c_s) + 1)}{2 n^{1/10000}} \E[\D] \le \frac{\E[\D]}{4}.\]
Similarly, $\frac{\E\left[\D_{n - c_s:n - c_s} \mid \D_{n - c_s:n - c_s} > \alpha_{n^{1/10000}}^{(\D_{n - c_s:n - c_s})}\right]}{n^{1/10000}} \le\frac{\E[\D]}{4}$, for an overall bound of $R_{\Linear_\gamma}(\D,\std^* \mid \overline{\Event^*}) \cdot \Pr[\overline{\Event^*}] \le 2 (\frac{\E[\D]}{4} + \frac{\E[\D]}{4}) =\E[\D]$. Putting everything together, we have
\begin{align*}
    R_{\Linear_\gamma}(\D,\std^* ) &= R_{\Linear_\gamma}(\D,\std^* \mid \Event^*) \cdot \Pr[\Event^*] + R_{\Linear_\gamma}(\D,\std^* \mid \overline{\Event^*}) \cdot \Pr[\overline{\Event^*}] \\
    &\le^{\text{(\Cref{lem:linear-upper-regime})}} 7 \E[\D] + \E[\D] \\
    &= 8 \E[\D]. \qedhere
\end{align*}
\end{proof}

\begin{proof}[Proof for \Cref{thr:master-lower-bound}]
    Combining \Cref{lem:optimal-lower} and \Cref{lem:linear-upper} gives us the result.
\end{proof}

\section{A threshold algorithm for selecting the best box}
\label{sec:positive}

In this section, we propose a new policy, $\IgnoreLarge$, and give sufficient conditions under which $\IgnoreLarge$'s expected reward is at most a constant factor of the expected reward of a prophet who knows $x_1, \dots, x_n$.

We will describe two versions of this policy. The first version works for all distributions; the second one is a slight modification that works for MHR distributions, under a weaker condition on the instance. Without loss of generality, we will assume that boxes are ordered in increasing $\sigma_i$, that is, $\sigma_1 \le \sigma_2 \le \dots \le \sigma_n$. 

\begin{itemize}[leftmargin=*]
    \item $\IgnoreLarge$: Pick $\alpha \in [0, 1]$ uniformly at random. Return $\argmax_{1 \le i \le \alpha n} y_i$.
    \item $\IgnoreLargeExp$: Pick $\alpha \in [0, 1]$ uniformly at random. Return $\argmax_{1 \le i \le n^\alpha} y_i$.
\end{itemize}

In~\Cref{thm:ignore-large} we present our guarantee for arbitrary distributions. Intuitively, if there is a universal constant $c$, e.g. $c=0.01$, such that a $c$ fraction of boxes have bounded noise (and specifically, $\sigma_i$ at most $\frac{\E[\D_{cn:cn}]}{5 \sqrt{2 \ln n}}$), then our policy gives a constant approximation to the reward of a prophet.

\begin{theorem}
\label{thm:ignore-large}
For all $c \in (0,1]$, for all distributions $\D$, all $n \geq 4$, and all $\std \in \good{\D}{n,c}$, we have
\[ R_{\IgnoreLarge}(\D, \std) \geq \frac{c^2}{20} \cdot \E[\D_{n:n}] \]
\end{theorem}

\begin{proof}
Consider $\std = (\sigma_1, \sigma_2, \dots, \sigma_n) \in \good{\D}{n,c}$ where, without loss of generality, we have $\sigma_1 \le \sigma_2 \le \dots \le \sigma_n$. As $\std \in \good{\D}{n,c}$, we have $\sigma_{cn} \le \frac{\E[\D_{cn:cn}]}{5 \sqrt{2 \ln n}}$.

Consider the event that $|\epsilon_i| \le \sigma_i \sqrt{2 \ln n}$ for all $1 \le i \le cn$. For any such box $i$, we have
\begin{align*}
\Pr\left[|\epsilon_i| \le \sigma_i \sqrt{2 \ln n}\right] &= \Pr\left[ |\Norm(0, \sigma_i^2)| \le \sigma_i \sqrt{2 \ln n} \right] \\
    &= 2 \Phi \left(\sqrt{2 \ln n} \right) - 1 \\
&\ge^{\text{(\Cref{lem:normal-bound})}} 2  \left( 1 - \frac{1}{\sqrt{2 \pi}} \frac{1}{\sqrt{2 \ln n}}\exp\left( -\frac{1}{2} \cdot 2 \ln n \right) \right) - 1 \\
    &= 1 - \frac{1}{n \sqrt{\pi \ln n}},
\end{align*}
and therefore
\[\Pr\left[|\epsilon_i| \le \sigma_i \sqrt{2 \ln n}, \forall i \in [1, cn]\right] \ge \left(1 - \frac{1}{n \sqrt{\pi \ln n}}\right)^{cn} \ge^{\text{(Bernoulli's inequality)}} 1 - \frac{c}{\sqrt{\pi \ln n}} \ge \frac{1}{2},\]
where the last inequality holds for all $n \geq 4 \geq e^{\frac{4c^2}{\pi}}$. Observe that, since $\sigma_i \le \frac{\E[\D_{cn:cn}]}{5\sqrt{2\ln n}}$ for all $i \in [1, cn]$, we can conclude that $\Pr[\max_{i \in [1, cn]} |\epsilon_i| \le \frac{1}{5} \cdot \E[\D_{cn:cn}]] \ge \frac{1}{2}$. Conditioned on this event we have $x_i - \frac{1}{5} \cdot \E[\D_{cn:cn}] \le y_i \le x_i + \frac{1}{5} \cdot \E[\D_{cn:cn}]$ for all $i \in [1, cn]$; therefore, for all $k \le cn$, we have $\max_{i \in [1, k]} y_i \ge \max_{i \in [1, k]} x_i - \frac{2}{5} \cdot \E[\D_{cn:cn}]$

We analyze the performance of $\IgnoreLarge$ under this event. Recall that $\IgnoreLarge$ draws $\alpha \in [0, 1]$ uniformly at random in its sampling step, and then outputs $\argmax_{i \in [1, \alpha n]} y_i$. There are two cases for $\alpha$:
\begin{itemize}[leftmargin=*]
    \item If $\alpha > c$, we will lower bound the expected reward of $\IgnoreLarge$ by $0$.
    \item If $\alpha \le c$, $\IgnoreLarge$ is going to pick the box with the largest $y_i$ among the first $\alpha n$ boxes. By our observation, $\IgnoreLarge$'s reward in this case is at least $\max_{i \in [1, \alpha n]} x_i - \frac{2}{5} \cdot \E[\D_{cn:cn}]$, and therefore the expected reward of $\IgnoreLarge$ in this case is at least
    \[E[\D_{\alpha n : \alpha n}] - \frac{2}{5} \cdot \E[\D_{cn:cn}] \ge^\text{(\Cref{lem:compare-order-stats})} \frac{\alpha}{c} \E[\D_{cn:cn}] - \frac{2}{5} \cdot \E[\D_{cn:cn}].\]
\end{itemize}
Therefore, conditioned on the event that $\max_{i \in [1, cn]} |\epsilon_i| \le \frac{1}{5} \cdot \E[\D_{cn:cn}]$, $\IgnoreLarge$'s expected reward is lower bounded by
\[
    \int_{\alpha = 0}^{c} \frac{\alpha}{c} \E[\D_{cn:cn}] - \frac{2}{5} \cdot \E[\D_{cn:cn}] \, d \alpha = \frac{c}{10} \cdot \E[\D_{cn:cn}].
\]

When this event does not occur, we lower bound $\IgnoreLarge$'s expected reward by $0$. Combining everything together, $\IgnoreLarge$'s expected reward is
\[R_{\IgnoreLarge}(\D, \std) \ge \frac{1}{2} \cdot \frac{c}{10} \cdot \E[\D_{cn:cn}] \ge^\text{(\Cref{lem:compare-order-stats})} \frac{c^2}{20} \cdot \E[\D_{n:n}]. \qedhere \]
\end{proof}

In~\Cref{thm:ignore-large-expo-mhr} we present an analog to~\Cref{thm:ignore-large} for MHR distributions. Here, our condition for getting a constant approximation is a lot weaker. Intuitively, if there is a universal constant $c$, such that $n^c$ boxes have bounded noise (and specifically, $\sigma_i$ at most $\frac{\E[\D_{cn:cn}]}{18 \sqrt{2c \ln n}}$), then our policy gives a constant approximation to the reward of a prophet.

\begin{theorem}
\label{thm:ignore-large-expo-mhr}
For all $c \in (0,1]$, for all MHR distributions $\D$, all $n \geq e^{\frac{4}{c \pi}}$, and all $\std \in \goodMHR{\D}{n,c}$, we have
\[ R_{\IgnoreLargeExp}(\D, \std) \geq \frac{c^2}{576} \cdot \E[\D_{n:n}]. \]
\end{theorem}

\noindent The proof of~\Cref{thm:ignore-large-expo-mhr} follows a similar structure to the proof of~\Cref{thm:ignore-large} and is deferred to~\Cref{app: sec positive}.
\bibliographystyle{alpha}
\bibliography{refs.bib}

\begin{thebibliography}{CAMTM20}

\bibitem[AB20]{allouah2020prior}
Amine Allouah and Omar Besbes.
\newblock Prior-independent optimal auctions.
\newblock {\em Management Science}, 66(10):4417--4432, 2020.

\bibitem[BBC11]{bertsimas2011theory}
Dimitris Bertsimas, David~B Brown, and Constantine Caramanis.
\newblock Theory and applications of robust optimization.
\newblock {\em SIAM review}, 53(3):464--501, 2011.

\bibitem[BKMR12]{bax2012comparing}
Eric Bax, Anand Kuratti, Preston Mcafee, and Julian Romero.
\newblock Comparing predicted prices in auctions for online advertising.
\newblock {\em International Journal of Industrial Organization}, 30(1):80--88,
  2012.

\bibitem[BMP19]{braverman2019sorted}
Mark Braverman, Jieming Mao, and Yuval Peres.
\newblock Sorted top-k in rounds.
\newblock In {\em Conference on Learning Theory}, pages 342--382. PMLR, 2019.

\bibitem[BMW16]{braverman2016parallel}
Mark Braverman, Jieming Mao, and S~Matthew Weinberg.
\newblock Parallel algorithms for select and partition with noisy comparisons.
\newblock In {\em Proceedings of the forty-eighth annual ACM symposium on
  Theory of Computing}, pages 851--862, 2016.

\bibitem[BP96]{barlow1996}
Richard~E. Barlow and Frank Proschan.
\newblock {\em Mathematical Theory of Reliability}.
\newblock Society for Industrial and Applied Mathematics, 1996.

\bibitem[CAMTM20]{cohen2020instance}
Vincent Cohen-Addad, Frederik Mallmann-Trenn, and Claire Mathieu.
\newblock Instance-optimality in the noisy value-and comparison-model* accept,
  accept, strong accept: Which papers get in?
\newblock In {\em Proceedings of the Fourteenth Annual ACM-SIAM Symposium on
  Discrete Algorithms}, pages 2124--2143. SIAM, 2020.

\bibitem[CD11]{cai2011extreme}
Yang Cai and Constantinos Daskalakis.
\newblock Extreme-value theorems for optimal multidimensional pricing.
\newblock In {\em 2011 IEEE 52nd Annual Symposium on Foundations of Computer
  Science}, pages 522--531, 2011.

\bibitem[CR14]{cole2014sample}
Richard Cole and Tim Roughgarden.
\newblock The sample complexity of revenue maximization.
\newblock In {\em Proceedings of the forty-sixth annual ACM symposium on Theory
  of computing}, pages 243--252, 2014.

\bibitem[DRY10]{dhangwatnotai2010revenue}
Peerapong Dhangwatnotai, Tim Roughgarden, and Qiqi Yan.
\newblock Revenue maximization with a single sample.
\newblock In {\em Proceedings of the 11th ACM conference on Electronic
  commerce}, pages 129--138, 2010.

\bibitem[DW12]{daskalakis2012symmetries}
Constantinos Daskalakis and Seth~Matthew Weinberg.
\newblock Symmetries and optimal multi-dimensional mechanism design.
\newblock In {\em Proceedings of the 13th ACM conference on Electronic
  commerce}, pages 370--387, 2012.

\bibitem[FRPU94]{feige1994computing}
Uriel Feige, Prabhakar Raghavan, David Peleg, and Eli Upfal.
\newblock Computing with noisy information.
\newblock {\em SIAM Journal on Computing}, 23(5):1001--1018, 1994.

\bibitem[GHZ19]{guo2019settling}
Chenghao Guo, Zhiyi Huang, and Xinzhi Zhang.
\newblock Settling the sample complexity of single-parameter revenue
  maximization.
\newblock In {\em Proceedings of the 51st Annual ACM SIGACT Symposium on Theory
  of Computing}, pages 662--673, 2019.

\bibitem[GJZ21]{guo2021robust}
Wenshuo Guo, Michael Jordan, and Emmanouil Zampetakis.
\newblock Robust learning of optimal auctions.
\newblock {\em Advances in Neural Information Processing Systems},
  34:21273--21284, 2021.

\bibitem[Gor41]{gordon1941}
Robert~David Gordon.
\newblock Values of mills' ratio of area to bounding ordinate and of the normal
  probability integral for large values of the argument.
\newblock {\em Annals of Mathematical Statistics}, 12:364--366, 1941.

\bibitem[GPZ21]{giannakopoulos2021optimal}
Yiannis Giannakopoulos, Diogo Po{\c{c}}as, and Keyu Zhu.
\newblock Optimal pricing for mhr and $\lambda$-regular distributions.
\newblock {\em ACM Transactions on Economics and Computation (TEAC)},
  9(1):1--28, 2021.

\bibitem[HMR15]{huang2015making}
Zhiyi Huang, Yishay Mansour, and Tim Roughgarden.
\newblock Making the most of your samples.
\newblock In {\em Proceedings of the Sixteenth ACM Conference on Economics and
  Computation}, pages 45--60, 2015.

\bibitem[HR09]{hartline2009simple}
Jason~D Hartline and Tim Roughgarden.
\newblock Simple versus optimal mechanisms.
\newblock In {\em Proceedings of the 10th ACM conference on Electronic
  commerce}, pages 225--234, 2009.

\bibitem[LNN61]{leone1961folded}
Fred~C Leone, Lloyd~S Nelson, and RB~Nottingham.
\newblock The folded normal distribution.
\newblock {\em Technometrics}, 3(4):543--550, 1961.

\bibitem[LP18]{liu2018competition}
Siqi Liu and Christos-Alexandros Psomas.
\newblock On the competition complexity of dynamic mechanism design.
\newblock In {\em Proceedings of the Twenty-Ninth Annual ACM-SIAM Symposium on
  Discrete Algorithms}, pages 2008--2025. SIAM, 2018.

\bibitem[MMW22]{mahdian2022regret}
Mohammad Mahdian, Jieming Mao, and Kangning Wang.
\newblock Regret minimization with noisy observations.
\newblock {\em arXiv preprint arXiv:2207.09435}, 2022.

\bibitem[Tha88]{thaler1988anomalies}
Richard~H Thaler.
\newblock Anomalies: The winner's curse.
\newblock {\em Journal of economic perspectives}, 2(1):191--202, 1988.

\end{thebibliography}

\appendix

\section{A technical lemma}

The following technical lemma will be useful throughout this appendix.

\begin{lemma}\label{lem: monotonicity of expected posterior}
For a random variable $Y = X + \epsilon$, where $\epsilon \sim \Norm(0, \sigma^2)$, it holds that $\E[X \mid Y = y]$ is monotone non-decreasing in $y$. 
\end{lemma}

\begin{proof}[Proof of \Cref{lem: monotonicity of expected posterior}]
Let $A(y) = \int_{0}^{\infty} x \cdot f(x) \cdot f_\Norm(y - x) \, dx$ and $B(y) = \int_{0}^{\infty} f(x) \cdot f_\Norm(y - x) \, dx$, then $\E[X \mid Y = y] = \frac{A(y)}{B(y)}$. We first compute the derivative of $f_\Norm(y - x)$:
\begin{align*}
    \frac{df_\Norm(y - x) }{dy} &= \frac{d}{dy} \left(\frac{1}{\sigma \sqrt{2 \pi}} \exp\left( -\frac{1}{2} \cdot \left(\frac{y - x}{\sigma} \right)^2 \right) \right) \\
    &= \frac{1}{\sigma \sqrt{2 \pi}} \exp\left(-\frac{1}{2} \cdot \left(\frac{y - x}{\sigma} \right)^2\right) \cdot \frac{x - y}{\sigma^2} \\
    &= f_\Norm(y - x) \cdot \frac{x - y}{\sigma^2}.
\end{align*}

Let $C(y) = \int_{0}^{\infty} x^2 \cdot f(x) \cdot f_\Norm(y - x) \, dx$. The derivative for $A(y)$ is
\begin{align*}
    \frac{dA(y)}{dy} &= \frac{d}{dy} \left(\int_{0}^{\infty} x \cdot f(x) \cdot f_\Norm(y - x) \, dx\right) \\
    &= \int_{0}^{\infty} x \cdot f(x) \cdot f_\Norm(y - x) \cdot \frac{x - y}{\sigma^2}\, dx \\
    &= \frac{1}{\sigma^2} \left(C(y) - y \cdot A(y)\right).
\end{align*}

The derivative for $B(y)$ is
\begin{align*}
\frac{dB(y)}{dy} &= \frac{d}{dy} \left(\int_{0}^{\infty} f(x) \cdot f_\Norm(y - x) \, dx\right) \\
        &= \int_{0}^{\infty} f(x) \cdot f_\Norm(y - x) \cdot \frac{x - y}{\sigma^2}\, dx \\
        &= \frac{1}{\sigma^2} \left(A(y) - y \cdot B(y)\right)
\end{align*}

Finally, the derivative for $\E[X \mid Y=y]$ is

\begin{align*}
\frac{d}{dy} \E[X \mid Y = y] &= \frac{d}{dy} \frac{A(y)}{B(y)} \\
        &= \frac{\frac{dA(y)}{dy} \cdot B(y) - \frac{dB(y)}{dy} \cdot A(y)}{B(y)^2} \\
        &= \frac{ \left(  \frac{1}{\sigma^2} \left(C(y) - y A(y)\right) \right) \cdot B(y) - \left( \frac{1}{\sigma^2} \left(A(y) - y B(y)\right) \right) \cdot A(y)}{B(y)^2} \\
        &= \frac{B(y)C(y) - y A(y) B(y) - A(y)^2 + yA(y)B(y)}{(\sigma B(y))^2} \\
        &= \frac{B(y)C(y) - A(y)^2}{(\sigma B(y))^2}.
\end{align*}
Since $\left(x \cdot f(x) \cdot f_\Norm(y - x)\right)^2 = \left(f(x) \cdot f_\Norm(y - x)\right) \cdot \left(x^2 \cdot f(x) \cdot f_\Norm(y - x)\right)$, the Cauchy-Schwarz inequality implies that $B(y)C(y) \geq A(y)^2$.
Therefore $\frac{d}{dy} \E[X \mid Y = y] = \frac{B(y)C(y) - A(y)^2}{(\sigma B(y))^2} \geq 0$.
\end{proof}

\section{Proofs missing from Section~\ref{subsec:technical}}\label{app: missing from technical}

\begin{proof}[Proof of~\Cref{lem:compare-order-stats}]
    It is sufficient to prove that $\frac{\E[\D_{\ell:\ell}]}{\ell} \ge \frac{\E[\D_{\ell+1:\ell+1}]}{\ell+1}$ for all integers $\ell \ge 1$. For all $t \in [0, 1]$, we have
    \begin{align*}
        \sum_{i = 0}^{\ell-1} t^i &\ge \ell t^\ell \\
        (1-t)\sum_{i = 0}^{\ell-1} t^i &\ge \ell (1-t) t^\ell \\
        1 - t^\ell &\ge \ell(t^\ell - t^{\ell+1}) \\
        \ell+1-(\ell+1)t^\ell &\ge \ell - \ell t^{\ell+1} \\
        \frac{1 - t^\ell}{\ell} &\ge \frac{1 - t^{\ell+1}}{\ell+1}
    \end{align*}

    Substituting $t = F(x)$ and taking integrals on both sides, we get $\frac{\int_{0}^{\infty} 1 - F(x)^\ell}{\ell} \, dx \ge  \frac{\int_{0}^{\infty} 1 - F(x)^{\ell + 1}}{n + 1} \, dx$, which proves our statement.
\end{proof}


\paragraph{Lemmas about MHR distributions}

We will heavily use the fact that order statistics of MHR distributions are also MHR (Theorem 5.5 on page 39 of~\cite{barlow1996}):

\begin{lemma}[\cite{barlow1996}]
\label{lem:mhr-order-stat}
    For any MHR\footnote{\cite{barlow1996} use the term IFR (increasing failure rate).\label{ftn:mhr-ifr}} random variable $X$ and any integers $1 \le k \le n$, $X_{k:n}$ is also MHR.
\end{lemma}

\begin{proof}[Proof of~\Cref{lem:order-stat-vs-quantile}]
Define $\zeta_p^{(\D)} = \inf\{x \mid F(x) \ge p\}$ as the $p$-th quantile of $\D$.

For the lower bound, we first observe that $\Pr[\D_{n:n} \le \alpha^{(\D)}_n] = \Pr[\D \le \alpha^{(\D)}_n]^n = (1 - \frac{1}{n})^n$, where with $n \ge 4$ we get $\frac{81}{256} \le (1 - \frac{1}{n})^n \le \frac{1}{e}$. Therefore, $\zeta_{81/256}^{(\D_{n:n})} \le \alpha_n \le \zeta_{1/e}^{(\D_{n:n})}$.

We use the following result from~\cite{barlow1996} (Theorem 4.6 on page 30):
\begin{lemma}[\cite{barlow1996}]
\label{lem:barlow-quantile-bounds}
    Assume $X$ is MHR\footnotemark[1] with mean $\mu_1$. If $p \le 1 - 1/e$, then $-\ln(1 - p) \cdot \mu_1 \le \zeta_p^{X} \le -\frac{\ln(1 - p)}{p} \cdot \mu_1$.
\end{lemma}

From \Cref{lem:mhr-order-stat}, we know that $\D_{n:n}$ is also MHR. Since $\frac{81}{256} \leq 1/e \le 1 - 1/e$, we can invoke \Cref{lem:barlow-quantile-bounds} on $\zeta^{(\D_{n:n})}_{81/256}$ and $\zeta^{(\D_{n:n})}_{1/e}$. For the lower bound we have

\[
\alpha^{(\D)}_n \ge \zeta_{81/256}^{(\D_{n:n})} \ge -\ln(1 - 81/256) \cdot \E[\D_{n:n}] \ge \frac{1}{3} \cdot \E[\D_{n:n}].
\]

For the upper bound we have

\[\alpha^{(\D)}_n \le \zeta_{1/e}^{(\D_{n:n})} \le -\frac{\ln(1 - 1/e)}{1/e} \cdot \E[\D_{n:n}] \le \frac{5}{4} \cdot \E[\D_{n:n}]. \qedhere\]
\end{proof}

\section{Proofs missing from Section~\ref{sec: lower bounds for large noise}}\label{app: missing from lower bounds for regimes}

\begin{proof}[Proof of~\Cref{lem:order-stats-half-norm}]
    $\E[\D] = \sqrt{2/\pi}$ is a standard property to the half-normal distribution (and can also be confirmed by computing the mean of a folded-normal with parameter $\mu = 0$~\cite{leone1961folded}).

    For the MHR property, it suffices to show that $\frac{f_\D(x)}{1 - F_\D(x)}$ is an increasing function. Note that its derivative is $\frac{f'_\D(x)(1 - F_\D(x)) + f_\D^2(x)}{(1 - F_\D(x))^2}$, so we need the numerator to be non-negative.

    As $f_\D(x) = \sqrt{\frac{2}{\pi}} \exp\left(\frac{-x^2}{2}\right) = 2 \phi(x)$ and $F_\D(x) = \erf\left(\frac{x}{\sqrt{2}}\right) = 2 \Phi(x) - 1$, the numerator is
    \[f'_\D(x)(1 - F_\D(x)) + f_\D^2(x) = -2x\phi(x) (2 - 2\Phi(x)) + 4 \phi^2(x) = 4\phi(x) \left(\phi(x) - x(1 - \Phi(x))\right),\]
    where the last quantity is non-negative as $\phi(x) \ge 0$ and by \Cref{lem:normal-bound}, proving our claim.

    Finally, since $\D$ is MHR, we use results from~\Cref{subsec:technical} to bound $\E[\D_{n:n}]$. Observe that
    \begin{align*}
        F_\D(\sqrt{\ln n}) &= 2 \Phi(\sqrt{\ln n} - 1) \\
            &\le^\text{(\Cref{lem:normal-bound})} 2\left(1 - \frac{1}{\sqrt{2 \pi}} \frac{\sqrt{\ln n}}{1 + \ln n} \exp\left(-\frac{1}{2} \cdot \ln n\right)\right) - 1 \\
            &= 1 - \sqrt\frac{2}{\pi} \cdot \frac{\sqrt{\ln n}}{n^{1/2} (1 + \ln n)} \\
            &\le 1 - \frac{1}{n},
    \end{align*}
    where the last inequality holds for all $n \geq 8$.
    Therefore, $\alpha_n^{(\D)} \ge \sqrt{\ln n}$, which implies $\E[\D_{n:n}] \ge^\text{(\Cref{lem:order-stat-vs-quantile})} \frac{4}{5} \sqrt{\ln n}$. Similarly,
    \begin{align*}
        F_\D(\sqrt{2 \ln n}) &= 2 \Phi(\sqrt{2 \ln n} - 1) \\
            &\ge^\text{(\Cref{lem:normal-bound})} 2\left(1 - \frac{1}{\sqrt{2 \pi}} \frac{1}{\sqrt{2 \ln n}} \exp\left(-\frac{1}{2} \cdot 2 \ln n\right)\right) - 1 \\
            &= 1 - \sqrt\frac{2}{\pi} \cdot \frac{1}{n\sqrt{2 \ln n}} \\
            &\ge 1 - \frac{1}{n}.
    \end{align*}
    Therefore, $\alpha_n^{(\D)} \le \sqrt{2 \ln n}$, which means $\E[\D_{n:n}] \le^\text{(\Cref{lem:order-stat-vs-quantile})} 3\sqrt{2} \sqrt{\ln n}$.
    
\end{proof}

\begin{proof}[Proof of \Cref{lem:expected-posterior-half-norm}]
We have $\E[X_i \mid Y_i = y_i] = \displaystyle\frac{\int_{0}^{\infty} x \cdot f_\D(x) \cdot f_{\Norm(0, \sigma_i^2)}(y_i - x) \, dx}{\int_{0}^{\infty} f_\D(x) \cdot f_{\Norm(0, \sigma_i^2)}(y_i - x) \, dx}$. We first transform the numerator.
\begin{align*}
    \int_{0}^{\infty} f_\D(x) \cdot f_{\Norm(0, \sigma_i^2)}(y_i - x) \, dx &= \int_{0}^{\infty} \frac{\sqrt{2}}{\sqrt{\pi}} \exp\left(\frac{-x^2}{2}\right) \cdot \frac{1}{\sigma_i \sqrt{2 \pi}} \exp\left(\frac{-(y_i - x)^2}{2\sigma_i^2}\right) \, dx\\
    &= \frac{1}{\sigma_i \pi} \int_{0}^{\infty} \exp\left(-\frac{1}{2} \left(x^2 + \left(\frac{y_i-x}{\sigma_i}\right)^2\right)\right) \, dx
\end{align*}

Let's focus on $x^2 + \left(\frac{y_i-x}{\sigma_i}\right)^2$:
\begin{align*}
    x^2 + \left(\frac{y_i-x}{\sigma_i}\right)^2 &= \frac{(x\sigma_i)^2 + y_i^2 - 2 y_i x + x^2}{\sigma_i^2} \\
    &= \frac{\left(x \sqrt{\sigma_i^2 + 1}\right)^2 - 2 y_i x + y_i^2}{\sigma_i^2} \\
    &=^\text{(let $\lambda = \sqrt{\sigma_i^2 + 1}$)} \frac{\left(\lambda x \right)^2 - 2 \frac{y_i}{\lambda} \cdot \lambda x + \left(\frac{y_i}{\lambda}\right)^2 + y_i^2\left(1 - \frac{1}{\lambda^2}\right)}{\sigma_i^2} \\
    &=^\text{(let $\rho = \frac{y_i^2\left(1 - \frac{1}{\lambda^2}\right)}{\sigma_i^2}$)} \left(\frac{\lambda x - \frac{y}{\lambda}}{\sigma_i}\right)^2 + \rho
\end{align*}

Observe that $\lambda$ and $\rho$ only depends on $\sigma_i$ and $y_i$. Therefore, coming back to the previous integral:
\begin{align*}
    \int_{0}^{\infty} f_\D(x) \cdot f_{\Norm(0, \sigma_i^2)}(y_i - x) \, dx &= \frac{1}{\sigma_i \pi} \int_{0}^{\infty} \exp\left(-\frac{1}{2} \left(\left(\frac{\lambda x - \frac{y}{\lambda}}{\sigma_i}\right)^2 + \rho\right)\right) \, dx \\
    &= \frac{e^{-\rho / 2}\lambda \sqrt{2}}{\sqrt{\pi}}\int_{0}^{\infty} \frac{1}{\sqrt{2 \pi} \cdot \lambda \sigma_i}\exp\left(-\frac{1}{2} \left(\frac{x - \frac{y_i}{\lambda^2}}{\lambda \sigma_i}\right)^2\right) \, dx \\
    &= \frac{e^{-\rho / 2}\lambda \sqrt{2}}{\sqrt{\pi}}\int_{0}^{\infty} f_{\Norm\left(\frac{y_i}{\lambda^2}, \left(\frac{\sigma_i}{\lambda}\right)^2\right)}(x) \, dx
\end{align*}

Calculated similarly, we have
\begin{align*}
    \int_{0}^{\infty} x \cdot f_\D(x) \cdot f_{\Norm(0, \sigma_i^2)}(y_i - x) \, dx &= \frac{e^{-\rho / 2}\lambda \sqrt{2}}{\sqrt{\pi}}\int_{0}^{\infty} x \cdot f_{\Norm\left(\frac{y_i}{\lambda^2}, \left(\frac{\sigma_i}{\lambda}\right)^2\right)}(x) \, dx
\end{align*}

Therefore
\begin{align*}
    \E[X_i \mid Y_i = y_i] &= \frac{\int_{0}^{\infty} x \cdot f_\D(x) \cdot f_{\Norm(0, \sigma_i^2)}(y_i - x) \, dx}{\int_{0}^{\infty} f_\D(x) \cdot f_{\Norm(0, \sigma_i^2)}(y_i - x) \, dx} \\
        &= \frac{\frac{e^{-\rho / 2}\lambda \sqrt{2}}{\sqrt{\pi}}\int_{0}^{\infty} x \cdot f_{\Norm\left(\frac{y_i}{\lambda^2}, \left(\frac{\sigma_i}{\lambda}\right)^2\right)}(x) \, dx}{\frac{e^{-\rho / 2}\lambda \sqrt{2}}{\sqrt{\pi}}\int_{0}^{\infty} f_{\Norm\left(\frac{y_i}{\lambda^2}, \left(\frac{\sigma_i}{\lambda}\right)^2\right)}(x) \, dx} \\
        &= \frac{\int_{0}^{\infty} x \cdot f_{\Norm\left(\frac{y_i}{\lambda^2}, \left(\frac{\sigma_i}{\lambda}\right)^2\right)}(x) \, dx}{\int_{0}^{\infty} f_{\Norm\left(\frac{y_i}{\lambda^2}, \left(\frac{\sigma_i}{\lambda}\right)^2\right)}(x) \, dx} \\
        &= \E\left[t \; \Big| \; t \sim \Norm\left(\frac{y_i}{\lambda^2}, \left(\frac{\sigma_i}{\lambda}\right)^2\right) \cap t \ge 0\right].
\end{align*}
This last quantity is the mean of the normal distribution $\Norm\left(\frac{y_i}{\sigma_i^2 + 1}, \left(\frac{\sigma_i}{\sqrt{\sigma_i^2 + 1}}\right)^2\right)$ truncated to $[0, \infty)$ (as $\lambda = \sqrt{\sigma_i^2 + 1}$). We can conclude that
\[\E[X_i \mid Y_i = y_i] = \frac{y_i}{\sigma_i^2 + 1} + \frac{\phi\left(\frac{-y_i}{\sigma_i \sqrt{\sigma_i^2 + 1}}\right)}{1 - \Phi\left(\frac{-y_i}{\sigma_i \sqrt{\sigma_i^2 + 1}}\right)} \cdot \frac{\sigma_i}{\sqrt{\sigma_i^2 + 1}}.\]
\end{proof}

\begin{proof}[Proof of~\Cref{thm: opt bad vs random}]
We follow the same proof structure as in~\Cref{thm:opt-bad-vs-prophet}.
Consider $\D = |\Norm(0, 1^2)|$. Consider $\std = (\sigma_1, \sigma_2, \dots, \sigma_n) \in \bad{\D}{n, c}$ where, without loss of generality, we have $\sigma_1 \le \sigma_2 \le \dots \le \sigma_n$. This means that $\sigma_{cn} > \frac{\E[\D_{cn:cn}] \cdot \sqrt{\ln(n)}}{\ln(cn)}$. Note that the expected reward of the optimal policy is at most the expected reward of the optimal policy that picks $2$ boxes $u$ and $v$ where $u \in [1, cn - 1]$ and $v \in [cn, n]$, and \emph{then enjoys the rewards of both boxes}. 

The expected reward from choosing box $u$ is at most $\E[\max_{i \in [1, cn - 1]} x_i] \le \E[\D_{cn:cn}]$. The expected reward from choosing box $v$ is at most the expected reward of $\Opt_{\D}$ conditioned on it choosing boxes from $cn$ to $n$, which in turn is at most $\max_{i \in [cn, n]} \E[X_i \mid Y_i = y_i]$. Therefore, the expected reward from box $v$ is upper bounded by:

\begin{align*}
    \E_{\observ}\left[\max_{i \in [cn, n]} \E[X_i \mid Y_i = y_i]\right] &\le^\text{(\Cref{lem:upp-bound-expected-posterior})} \E_{\observ}\left[\max_{i \in [cn, n]} U_{\sigma_i}(y_i)\right] \\
        &= \E\left[\max_{i \in [cn, n]} U_{\sigma_i}\left(X_i + \Norm(0, \sigma_i^2)\right)\right] \\
        &\le^\text{($U_{\sigma_i}(y)$ is monotone)} \E\left[\max_{i \in [cn, n]} U_{\sigma_i}\left(X_i + |\Norm(0, \sigma_i^2)|\right)\right] \\
        &= \E\left[\max_{i \in [cn, n]} \sqrt\frac{2}{\pi} + \frac{\left(X_i + |\Norm(0, \sigma_i^2)|\right)}{\sigma_i^2 + 1}\right] \\
        &\le \E\left[\sqrt\frac{2}{\pi} + \max_{i \in [cn, n]} \frac{X_i}{\sigma^2_{i}} + \max_{i \in [cn, n]} \frac{|\Norm(0, \sigma_i^2)|}{\sigma_i^2}\right] \\
        &\le \sqrt\frac{2}{\pi} + \frac{\E\left[ |\Norm(0, 1)|_{n:n} \right]}{\sigma^2_{cn}}  + \E\left[ \max_{i \in [cn, n]} \left|\Norm\left(0, \frac{1}{\sigma_i^2}\right)\right| \right] \\
        &\leq^\text{(\Cref{lem:order-stats-half-norm-general})} \sqrt\frac{2}{\pi} + \frac{ 3\sqrt{2} \cdot \sqrt{\ln n}}{\sigma_{cn}^2} + \frac{1}{\sigma_{cn}} \cdot 3\sqrt{2} \cdot \sqrt{\ln n} \\
        &\leq \E[\D] + \frac{6\sqrt{2} \cdot \sqrt{\ln n}}{\sigma_{cn}} \\
        &\leq^{\left( \sigma_{cn} > \frac{\E[\D_{cn:cn}] \cdot \sqrt{\ln(n)} }{\ln(cn)} \right)}
        \E[\D] + \frac{6\sqrt{2} \cdot \ln(cn)}{\E[\D_{cn:cn}]} \\
        &\leq^\text{(\Cref{lem:order-stats-half-norm})} \E[\D] + \frac{6\sqrt{2} \cdot \frac{25}{16} (\E[\D_{cn:cn}])^2}{\E[\D_{cn:cn}]}\\
        &\leq^{(c n \geq 1)} 15 \E[\D_{cn:cn}].
\end{align*}

Combining, we get $R_{\Opt_\D}(\D, \std) \le 16 \E[\D_{cn:cn}]$. Noting that, by~\Cref{lem:order-stats-half-norm}, $\E[\D_{cn:cn}] \leq 3\sqrt{2} \sqrt{\ln(cn)} = 3 \sqrt{\pi} \sqrt{\ln(cn)} \E[\D]$, we have $R_{\Opt_\D}(\D, \std) \le 16 \cdot 3 \sqrt{\pi} \sqrt{\ln(cn)} E[\D] \leq 86 \sqrt{\ln(cn)} E[\D]$, as desired.
\end{proof}

\section{Proofs missing from Section~\ref{sec:lower-bounds}}\label{app: missing from 3}

\subsection{Proofs missing from Section~\ref{subsec: naive fails}}

\begin{proof}[Proof of~\Cref{lem:large-noise-eps-bound-s}]

Formally, this event is $\max_{i \in [n - c_b + 1, n]} \epsilon_i > \beta_{n^2}^{(\D_{n:n})}$. We have

\begin{align*}
    \Pr\left[\max_{i \in [n - c_b + 1, n]} \epsilon_i > \beta_{n^2}^{(\D_{n:n})}\right] &= 1 - \Pr\left[\max_{i \in [n - c_b + 1, n]} \epsilon_i \le \beta_{n^2}^{(\D_{n:n})}\right] \\
    &= 1 - \Pr\left[\Norm(0, \sigma_b^2) \le \beta_{n^2}^{(\D_{n:n})}\right]^{c_b} \\
    &= 1 - \Pr\left[\Norm(0, \sigma_b^2) \le \frac{\sigma_b}{6 \sqrt{\ln n}}\right]^{c_b} \\
    &\ge 1 - \Pr\left[\Norm(0, \sigma_b^2) \le \frac{\sigma_b}{6}\right]^{6 \ln n}
\end{align*}
Using the fact that $\Pr\left[ \Norm(\mu, \sigma^2) \le x \right] = \Phi(\frac{x-\mu}{\sigma})$, where $\Phi(x) = \frac{1}{\sqrt{2\pi}} \int_{-\infty}^x e^{-t^2/2} dt$ is the CDF of the standard normal distribution, we have that 
$\Pr\left[\max_{i \in [n - c_b + 1, n]} \epsilon_i > \beta_{n^2}^{(\D_{n:n})}\right] \geq 1 - \Phi\left( \frac{1}{6} \right) ^{6 \ln n}$. Since $\Phi\left( \frac{1}{6} \right) < 0.6$  we have
\[
\Pr\left[\max_{i \in [n - c_b + 1, n]} \epsilon_i > \beta_{n^2}^{(\D_{n:n})}\right] \geq 1 - ((0.6)^2)^{3 \ln n} \geq 1 - \left( \frac{1}{e} \right)^{3 \ln n} \geq 1 - \frac{1}{n^3}. \qedhere
\]
\end{proof}

\begin{proof}[Proof of~\Cref{lemma:large-box-reward}]
    Note that as $\epsilon_i \sim \Norm(0, \sigma_b^2)$ and $\sigma_b = 6 \beta^{(\D_{n:n})}_{n^2} \sqrt{\ln n}$ we have
    \begin{align*}
    \Pr[\epsilon_i \le 12 \beta_{n^2}^{(\D_{n:n})} \ln n]& = \Pr[\epsilon_i \le 2 \sqrt{\ln n} \cdot \sigma_b)] \\
    &= \Phi(2 \sqrt{\ln n}) \\
    &\ge^{\text{(\Cref{lem:normal-bound})}} 1 - \frac{1}{\sqrt{2 \pi}} \frac{1}{2 \sqrt{\ln n}} \cdot \exp(-2 \ln n) \\
    &= 1 - \frac{1}{2 \sqrt{2 \pi}} \frac{1}{n^2 \sqrt{\ln n}} \\
    &\ge 1 - \frac{1}{n^2}.\qedhere
    \end{align*}

\end{proof}

\begin{proof}[Proof of~\Cref{lem:large-box-reward-bounded-D}]

Slightly overloading notation, let $f(x)$ be the PDF of $Z$. Let $A(y) = \int_{0}^{V} x \cdot f(x) \cdot f_\Norm(y - x) \, dx$ and $B(y) = \int_{0}^{V} f(x) \cdot f_\Norm(y - x) \, dx$, then $\E[Z \mid Z + \Norm(0,\sigma^2) = y] = \frac{A(y)}{B(y)}$.
From \Cref{lem: monotonicity of expected posterior} we know that $\E[Z \mid Z + \Norm(0,\sigma^2) = y]$ is monotone non-decreasing in $y$. 

Let $r = \frac{\sigma}{V}$. Consider $y^* = \frac{\sigma^2}{2V} = \sigma \cdot \frac{r}{2}$. As $\sigma > 2V$ or $r > 2$, we then have $y^* > \sigma > V$, which implies that $f_\Norm(y^* - V) \geq f_\Norm(y^* - x)$ for all $x \in [0, V]$. We then have the following bound on $A(y^*)$:
\begin{align*}
    A(y^*) &= \int_{0}^{V} x \cdot f(x) \cdot f_\Norm(y^* - x) \, dx \\
        &\le \int_{0}^{V} x \cdot f(x) \cdot f_\Norm(y^* - V) \, dx \\
        &= \E[Z] \cdot f_\Norm(y^* - V) \\
        &= \E[Z] \cdot \frac{1}{\sigma \sqrt{2 \pi}} \exp\left(- \frac{1}{2} \left( \frac{y^* - V}{\sigma} \right)^2 \right).
\end{align*}
Recalling that $y^* = \sigma \cdot \frac{r}{2}$ and that $V = \frac{\sigma}{r}$, we have:
\begin{align*}
        A(y^*) &= \frac{1}{\sigma \sqrt{2 \pi}} \E[Z] \cdot \exp\left(-\frac{1}{2} \left(\frac{r}{2} - \frac{1}{r} \right)^2 \right) \\
        &= \frac{1}{\sigma \sqrt{2 \pi}} \E[Z] \cdot \exp\left(-\frac{r^2}{8} + \frac{1}{2} - \frac{1}{2 r^2} \right) \\
        &\le \frac{1}{\sigma \sqrt{2 \pi}} \E[Z] \cdot \frac{\sqrt{e}}{\exp(r^2/8)}.
\end{align*}

Meanwhile, for $B(y^*)$, we have

\begin{align*}
    B(y^*) &= \int_{0}^{V} f(x) \cdot f_\Norm(y^* - x) \, dx \\
        &\ge^{\text{($y^* \geq V$)}} \int_{0}^{V} f(x) \cdot f_\Norm(y^*) \, dx \\
        &= f_\Norm(y^*) \cdot \int_{0}^{V} f(x) \, dx \\
        &= f_\Norm(y^*) \\
        &= \frac{1}{\sigma \sqrt{2 \pi}} \exp\left(- \frac{1}{2} \left(\frac{y^*}{\sigma} \right)^2 \right) \\
        &= \frac{1}{\sigma \sqrt{2 \pi}} \cdot \frac{1}{\exp(r^2/8)}.
\end{align*}

Therefore $A(y^*) \le 2 \E[Z] \cdot B(y^*)$, and thus $\E[Z \mid Z + \Norm(0,\sigma^2) = y^*] = \frac{A(y^*)}{B(y^*)}$ is at most $2 \E[Z]$. Since $\E[Z \mid Z + \Norm(0,\sigma^2)=y]$ is monotone non-decreasing in $y$ (\Cref{lem: monotonicity of expected posterior}), we can conclude that $\E[Z \mid Z + \Norm(0,\sigma^2)=y] \le 2 \E[Z]$ for all $y \le y^* = \frac{\sigma^2}{2V}$.
\end{proof}

\subsection{Proofs missing from Section~\ref{subsec: lower bound linear}}

\begin{proof}[Proof of \Cref{lem:order-stat-order-stat-bound}]
Since $n \ge 4$ we have that $n^a \ge 4$ for all $a \ge 1$. Therefore,
\begin{align*}
        \E[\D_{n^a:n^a}] &\le^{\text{(\Cref{lem:order-stat-vs-quantile})}} 3\alpha^{(\D)}_{n^a} \\
            &\le^{\text{(\Cref{lem: cai daskalakis alpha bound})}} 3a \cdot \alpha^{(\D)}_n \\
            &\le^{\text{(\Cref{lem:order-stat-vs-quantile})}} \frac{15a}{4} \cdot \E[\D_{n:n}] \\
            &< 4a \cdot \E[\D_{n:n}].\qedhere
\end{align*}
\end{proof}

\begin{proof}[Proof of~\Cref{lem:eps-small-bound}]
Observe that $\epsilon_i$ are values drawn from $\Norm(0, \sigma_s^2)$. We then have

\begin{align*}
    \Pr\left[ \max_{i \in [2, c_s + 1]} \epsilon_i \le \frac{\theta^* \sigma_s}{37} \right] &= \Pr\left[ \Norm(0, \sigma_s^2) \leq \frac{\theta^* \sigma_s}{37}\right]^{c_s}\\
    &= \Phi\left( \frac{\theta^*}{37} \right)^{n^{1/5626}} \\
    &\ge^{\text{(\Cref{lem:normal-bound})}} \left(1 - \frac{1}{\sqrt{2 \pi}} \frac{37\sqrt{2}}{\sqrt{\ln n}} \exp\left( -\frac{1}{2} \cdot \frac{1}{2738} \ln n \right) \right)^{n^{1/5626}} \\
    &\ge^{\text{(Bernoulli's inequality)}} 1 -\frac{37}{\sqrt{\pi} \sqrt{\ln n}} n^{1/5626-1/5476} \\
    &\ge 1 - \frac{1}{\ln n}.
\end{align*}
\end{proof}

\begin{proof}[Proof of~\Cref{lem:eps-large-bound}]
Observe that $\epsilon_i$ are values drawn from $\Norm(0, \sigma_b^2)$. We then have

\begin{align*}
    \Pr\left[ \max_{i \in [c_s + 2, n]} \epsilon_i - \theta^* \sigma_b \ge \sigma_b \right] &= 1 - \Pr\left[ \max_{i \in [c_s + 2, n]} \epsilon_i \le \theta^* \sigma_b + \sigma_b \right] \\
    &= 1 - \Pr\left[ \Norm(0,\sigma_b^2) \le \theta^* \sigma_b + \sigma_b \right]^{n - c_s - 1}\\
        &= 1 - (\Phi(\theta^* + 1))^{n - c_s - 1} \\
        &\ge 1 - \left( \Phi(\sqrt{2} \theta^*) \right)^{n/2} \\
    &\ge^{\text{(\Cref{lem:normal-bound})}} 1 - \left( 1 - \frac{1}{\sqrt{2 \pi}} \frac{\sqrt{2} \theta^*}{2 (\theta^*)^2 + 1} \exp(-(\theta^*)^2) \right)^{\frac{n}{2}} \\
        &\ge^{\text{(Bernoulli's inequality)}} 1 - \frac{1}{1 + \frac{n}{2} \frac{1}{\sqrt{2 \pi}} \frac{\sqrt{\ln n}}{\ln n + 1} \exp\left( - \frac{\ln n}{2} \right) } \\
        &= 1 - \frac{1}{1 + \frac{\sqrt{n}}{2} \frac{1}{\sqrt{2 \pi}} \frac{\sqrt{\ln n}}{\ln n + 1}} \\
        &\ge 1 - \frac{1}{\ln n}. \qedhere 
\end{align*}
\end{proof}

\begin{proof}[Proof of~\Cref{lem:large-box-reward-linear}]
    The proof is similar to that of~\Cref{lemma:large-box-reward}.
    Note that as $\epsilon_i \sim \Norm(0, \sigma_b^2)$ and $\sigma_b = 6 \beta^{(\D_{n:n})}_{n^2} \sqrt{\ln n}$ we have
    \begin{align*}
    \Pr[\epsilon_i \le 12 \beta_{n^2}^{(\D_{n:n})} \ln n]& = \Pr[\epsilon_i \le 2 \sqrt{\ln n} \cdot \sigma_b)] \\
    &= \Phi(2 \sqrt{\ln n}) \\
    &\ge^{\text{(\Cref{lem:normal-bound})}} 1 - \frac{1}{\sqrt{2 \pi}} \frac{1}{2 \sqrt{\ln n}} \cdot \exp(-2 \ln n) \\
    &= 1 - \frac{1}{2 \sqrt{2 \pi}} \frac{1}{n^2 \sqrt{\ln n}} \\
    &\ge 1 - \frac{1}{n^2}.\qedhere
    \end{align*}
\end{proof}

\begin{proof}[Proof of~\Cref{lem:c-large}]
Consider any $c \geq \theta^*$.
Observe that $Y_1 - c \sigma^*_1 = X_1 \ge 0$.
We show that conditioned on $\Event_1 \cap \Event^*$, we have $\max_{i \in [2, c_s + 1]} Y_i \le \theta^* \sigma_s$. We first note that from~\Cref{lem:lower bound on max passing expectation}, we have $\Pr[\D_{c_s:c_s} < 2 \E[\D_{c_s:c_s}]] \ge 1 - \frac{1}{c_s^{3/5}} = 1 - \frac{1}{n^{1/5626 \cdot 3/5}} > 1 - \frac{1}{n^{1/10000}}$. Therefore, by~\Cref{dfn: alpha_m}, $2 \E[\D_{c_s:c_s}] \geq \alpha_{n^{1/10000}}^{(\D_{c_s:c_s})}$. Then, conditioned on both $\Event_1$ and $\Event^*$, we have that for any small noise box $i$:
\[
    Y_i = X_i + \epsilon_i <^{\text{(\Cref{dfn:events})}} \alpha_{n^{1/10000}}^{(\D_{c_s:c_s})} + \frac{\theta^* \sigma_s}{37} 
        \le 2 \E[\D_{c_s:c_s}] + \frac{\theta^* \sigma_s}{37} 
        = \frac{36 \theta^* \sigma_s}{37} + \frac{\theta^* \sigma_s}{37} = \theta^* \sigma_s.
\]
Therefore, conditioned on $\Event_1$ and $\Event^*$, we have $\max_{i \in [2, c_s + 1]} Y_i - c \sigma^*_i \le \theta^* \sigma_s - c \sigma_s < 0$, i.e. $Y_1 - c \sigma^*_1 > Y_i - \max_{i \in [2, c_s + 1]} Y_i - c \sigma^*_i$ and hence $\LinearFixed_c$ does not choose any small box $i$.
\end{proof}

\begin{proof}[Proof of~\Cref{lem:c-small}]



Consider any $c \geq \theta^*$.
Observe that conditioned on $\Event_2'$, $\max_{i \in [c_s + 2, n]} Y_i - c \sigma^*_i \ge^{\text{(Dfn~\ref{dfn:events})}} \sigma_b$.

From~\Cref{lem:lower bound on max passing expectation}, we have $\Pr[\D_{c_s:c_s} < 2 \E[\D_{c_s:c_s}]] \ge 1 - \frac{1}{c_s^{3/5}} = 1 - \frac{1}{n^{1/5626 \cdot 3/5}} > 1 - \frac{1}{n^{1/10000}}$. Therefore, $2 \E[\D_{c_s:c_s}] \geq \alpha_{n^{1/10000}}^{(\D_{c_s:c_s})}$. Then, conditioned on $\Event^* \cap \Event_1$, we have that for all $i \in [2, c_s + 1]$:
\[
    Y_i = X_i + \epsilon_i < \alpha_{n^{1/10000}}^{(\D_{c_s:c_s})} + \frac{\theta^* \sigma_s}{37} 
        \le 2 \E[\D_{c_s:c_s}] + \frac{\theta^* \sigma_s}{37} 
        = \frac{36 \theta^* \sigma_s}{37} + \frac{\theta^* \sigma_s}{37} = \theta^* \sigma_s.
\]
Therefore,
\begin{align*}
    \max_{i \in [1, c_s + 1]}  Y_i - c \sigma^*_i &= \max\{Y_1,  \max_{i \in [2, c_s + 1]} Y_i - c \sigma_s\} \\
        &\le \max\{X_1,  \max_{i \in [2, c_s + 1]} Y_i \} \\
        &\le^{\text{(Dfn~\ref{dfn:events})}} \max\{\alpha_{n^{1/10000}}^{(\D_{n - c_s:n - c_s})}, \theta^* \sigma_s\} \\
        &< \sigma_b,
\end{align*}
where the last inequality follows from the facts that $\theta^* \sigma_s < \sigma_b$ (see \Cref{lem:compare-sigmas} in the appendix) and that $\alpha_{n^{1/10000}}^{(\D_{n - c_s:n - c_s})} < \sigma_b = 6 \alpha_{n^{1/10000}}^{(\D_{n - c_s:n - c_s})} \sqrt{\ln n}$. Therefore, $\max_{i \in [c_s + 2, n]} Y_i - c \sigma^*_i > \max_{i \in [1, c_s + 1]}  Y_i - c \sigma^*_i$, and so $\LinearFixed_c$ chooses a large noise box.
\end{proof}

\begin{lemma}
\label{lem:order-stat-mean-bound}
    For any MHR distribution $\D$ supported on $[0, \infty)$ and for all $n \ge 1$, we have
    \[\E[\D_{n:n}] \le (\ln n + 1) \cdot \E[\D].\]
\end{lemma}

\begin{proof}[Proof of \Cref{lem:order-stat-mean-bound}]

The lemma is an immediate consequence of the following result from~\cite{barlow1996} (Corollary 4.10 on page 33):

\begin{lemma}[\cite{barlow1996}]
If $X_i$, $i=1,\dots,n$, are MHR\footnotemark[1] random variables with mean $\mu_i$ and cdf $F_i(.)$, and $G_i(x) = 1 - \exp(-x/\mu_i)$, then:
\[
\int_{0}^{\infty} 1-\prod_{i=1}^n F_i(x) dx \leq \int_{0}^{\infty} 1-\prod_{i=1}^n G_i(x) dx.
\]
\end{lemma}

Applying this result for the case of $F(x) = F_i(x)$ for all $i$, we have that
\[
\E[\D_{n:n}] = \int_{0}^{\infty} 1-F^n(x) dx \leq \int_{0}^{\infty} 1-(1 - e^{-\frac{x}{\E[\D]}})^n dx = \E[\D] \sum_{i=1}^{n} \frac{1}{i}
\]

Using the fact that $\sum_{i=1}^{n} \frac{1}{i} \leq \ln(n) + 1$, we get the lemma. 
\end{proof}

\begin{lemma}
\label{lem:very specific lemma on 6 log n}
    For any $n \ge 1$ and $m \ge 2$, we have
    \[\E[\D_{n:n} \mid \D_{n:n} > \alpha_m^{(\D_{n:n})}] \cdot \Pr[\D_{n:n} > \alpha_m^{(\D_{n:n})}] \le \frac{15 (\ln m + \ln n + 1)\E[\D]}{2m}. \]
\end{lemma}

\begin{proof}[Proof of \Cref{lem:very specific lemma on 6 log n}]
    We use the following result from \cite{cai2011extreme} (Lemma 36):
    \begin{lemma}[\cite{cai2011extreme}]
    \label{lem:contribution-of-high-tail}
        For any MHR distribution $\D$ and any $m \ge 2$, we have
        \[\E[\D \mid \D \ge \alpha_{m}^{(\D)}] \cdot \Pr[\D \ge \alpha_{m}^{(\D)}] \le \frac{6 \alpha^{(\D)}_m}{m}.\]
    \end{lemma}

Since order statistics of MHR distributions are also MHR (\Cref{lem:mhr-order-stat}), $\D_{n:n}$ and $(\D_{n:n})_{m:m} = \D_{nm:nm}$ are MHR. Then, by \Cref{lem:order-stat-vs-quantile} we have that 
\begin{equation}\label{eq: bound on nm}
\alpha_{m}^{(\D_{n:n})} \le \frac{5}{4} \cdot \E[\D_{nm:nm}].
\end{equation}

Towards proving the lemma, we then get
    
\begin{align*}
    \E[\D_{n:n} \mid \D_{n:n} > \alpha_m^{(\D_{n:n})}] \cdot \Pr[\D_{n:n} > \alpha_m^{(\D_{n:n})}] &\le^{\text{(\Cref{lem:contribution-of-high-tail})}} \frac{6 \alpha_m^{(\D_{n:n})}}{m} \\
 &\le^{\text{(\Cref{eq: bound on nm})}} 6 \cdot \frac{5}{4} \cdot \frac{\E[\D_{nm:nm}]}{m} \\
    &\le^{\text{(\Cref{lem:order-stat-mean-bound})}} \frac{15(\ln(nm) + 1)}{2m} \cdot \E[\D] \\
    &= \frac{15(\ln(n) + \ln(m) + 1)}{2m} \cdot \E[\D]. \qedhere
\end{align*}
\end{proof}

\begin{lemma}\label{lem:compare-sigmas}
$\sigma_b > \theta^* \sigma_s$.
\end{lemma}

\begin{proof}[Proof of~\Cref{lem:compare-sigmas}]

From~\Cref{lem:mhr-order-stat}, we know that $\D_{a:a}$ is MHR for any $a \ge 1$. Then, by \Cref{lem:order-stat-vs-quantile} we have that 

\begin{equation}\label{eq: bound on n to the stuff}
\alpha_{n^{1/10000}}^{(\D_{n-c_s:n-c_s})} \geq \frac{1}{3} \cdot \E[\left( \D_{n-c_s:n-c_s} \right)_{n^{1/10000}:n^{1/10000}}]
\end{equation}

Towards proving~\Cref{lem:compare-sigmas}:

\begin{align*}
    \sigma_b &= 6 \alpha_{n^{1/10000}}^{(\D_{n-c_s:n-c_s})} \sqrt{\ln n} \\
    &\ge^{\text{(\Cref{eq: bound on n to the stuff})}} 6 \cdot \frac{1}{3} \E[\D_{(n - c_s) \cdot n^{1/10000}:(n - c_s) \cdot n^{1/10000}}] \sqrt{\ln n} \\
    &> \frac{5}{2} \E[\D_{(n - c_s) \cdot n^{1/10000}:(n - c_s) \cdot n^{1/10000}}] \\
    &>^{(c_s = n^{1/5626})} \frac{5}{2} \E[\D_{c_s:c_s}] \\
    &= \theta^* \sigma_s. \qedhere
\end{align*}
\end{proof}

\section{Proofs missing from Section~\ref{sec:positive}}\label{app: sec positive}

\begin{proof}[Proof of~\Cref{thm:ignore-large-expo-mhr}]
Consider $\std = (\sigma_1, \sigma_2, \dots, \sigma_n) \in \goodMHR{\D}{n}$ where, without loss of generality, we have $\sigma_1 \le \sigma_2 \le \dots \le \sigma_n$. As $\std \in \goodMHR{\D}{n}$, there exists a constant $c = c_{(\D,n)} \in (0, 1]$ such that $\sigma_{n^c} \le \frac{\E[\D_{n^c:n^c}]}{18\sqrt{2 c \ln n}}$.

Consider the event that $|\epsilon_i| \le \sigma_i \sqrt{2c \ln n}$ for all $1 \le i \le n^c$. Following the same analysis as the proof of \Cref{thm:ignore-large}, for any box $i \in [1, n^c]$, we have

\begin{align*}
\Pr\left[|\epsilon_i| \le \sigma_i \sqrt{2 c \ln n}\right] &= \Pr\left[|\epsilon_i| \le \sigma_i \sqrt{2 \ln n^c}\right] \\
&= \Pr\left[ |\Norm(0, \sigma_i^2)| \le \sigma_i \sqrt{2 \ln n^c} \right] \\
    &= 2 \Phi \left(\sqrt{2 \ln n^c} \right) - 1 \\
&\ge^{\text{(\Cref{lem:normal-bound})}} 2  \left( 1 - \frac{1}{\sqrt{2 \pi}} \frac{1}{\sqrt{2 \ln n^c}}\exp\left( -\frac{1}{2} \cdot 2 \ln n^c \right) \right) - 1 \\
    &= 1 - \frac{1}{n^c \sqrt{c \pi \ln n}},
\end{align*},
and therefore
\[
\Pr\left[|\epsilon_i| \le \sigma_i \sqrt{2c \ln n}, \forall i \in [1, n^c]\right] \ge \left(1 - \frac{1}{n^c \sqrt{c \pi \ln n}}\right)^{n^c} \ge^{\text{(Bernoulli's inequality)}} 1 - \frac{n^c}{n^c \sqrt{c \pi \ln n}} \ge \frac{1}{2},
\]
where the last inequality holds for all $n \geq e^{\frac{4}{c \pi}}$.

Since $\sigma_i \le \frac{\E[\D_{n^c:n^c}]}{18\sqrt{2c \ln n}}$ for all $i \in [1, n^c]$, we can conclude that $\Pr[\max_{i \in [1, n^c]} |\epsilon_i| \le \frac{1}{18} \cdot \E[\D_{n^c:n^c}]] \ge \frac{1}{2}$. Conditioned on this event, for all $i \in [1, n^c]$, we have $x_i - \frac{1}{18} \cdot \E[\D_{n^c:n^c}] \le y_i \le x_i + \frac{1}{18} \cdot \E[\D_{n^c:n^c}]$; therefore, for all $k \le n^c$, we have $\max_{i \in [1, k]} y_i \ge \max_{i \in [1, k]} x_i - \frac{1}{9} \cdot \E[\D_{n^c:n^c}]$.

We analyze the performance of $\IgnoreLargeExp$ conditioned on this event. Recall that $\IgnoreLargeExp$ draws $\alpha \sim U[0, 1]$, and then outputs $\argmax_{i \in [1, n^{\alpha}]} y_i$. We consider two cases for $\alpha$:
\begin{itemize}[leftmargin=*]
\item If $\alpha > c$, we will lower bound the expected reward of $\IgnoreLargeExp$ by $0$.
\item If $\alpha \le c$, $\IgnoreLargeExp$ is going to pick the box with the largest $y_i$ among the first $n^\alpha$ boxes. By our observation, $\IgnoreLargeExp$'s reward in this case is at least $\max_{i \in [1, n^\alpha]} x_i - \frac{1}{9} \cdot \E[\D_{n^c:n^c}]$, and therefore the expected reward of $\IgnoreLargeExp$ in this case is at least $\E[\D_{n^\alpha : n^\alpha}] - \frac{1}{9} \cdot \E[\D_{n^c:n^c}]$. By \Cref{lem:order-stat-order-stat-bound}, since $\frac{c}{\alpha} \le 1$, we have $\E[\D_{n^c:n^c}] \le \frac{4c}{\alpha} \cdot \E[\D_{n^\alpha : n^\alpha}]$. Continuing our derivation, the expected reward of $\IgnoreLargeExp$ is at least
\[\E[\D_{n^\alpha : n^\alpha}] - \cdot \E[\D_{n^c:n^c}] \ge \frac{\alpha}{4c} \cdot \E[\D_{n^c:n^c}] - \frac{1}{9} \cdot \E[\D_{n^c:n^c}].\]
\end{itemize}
Therefore, conditioned on the event that $\max_{i \in [1, n^c} |\epsilon_i| \le \frac{1}{18} \cdot \E[\D_{n^c:n^c}]$, $\IgnoreLargeExp$'s expected reward is lower bounded by
\[
    \int_{\alpha = 0}^{c} \frac{\alpha}{4c} \cdot \E[\D_{n^c:n^c}] - \frac{1}{9} \cdot \E[\D_{n^c:n^c}] \, d \alpha = \frac{1}{72} \cdot \E[\D_{n^c:n^c}].
\]

In outcomes outside this event, we can lower bound $\IgnoreLargeExp$'s expected reward by $0$. Combining everything, $\IgnoreLargeExp$'s expected reward is
\[R_{\IgnoreLargeExp}(\D, \std) \ge \frac{1}{2} \cdot \frac{1}{72} \cdot \E[\D_{n^c:n^c}] \ge^\text{(\Cref{lem:order-stat-order-stat-bound})} \frac{c^2}{576} \cdot \E[\D_{n:n}]. \qedhere \]
\end{proof}

\end{document}